\documentclass[11pt,journal,twocolumn]{IEEEtran}

\usepackage[margin=2.2cm]{geometry}
\usepackage[utf8]{inputenc} 
\usepackage{graphicx}
\usepackage[T1]{fontenc}
\usepackage{url}
\usepackage{hyperref}
\usepackage{xcolor}

\usepackage{ifthen}
\usepackage{subfig}
\usepackage{amssymb}
\usepackage{amsthm}
\usepackage{float}
\usepackage[cmex10]{amsmath} 
\usepackage[numbers]{natbib}                             

\newcommand{\ERM}{{\sf{ERM}}}

\newcommand{\E}[1]{\mathbb E\left[#1\right]}
\newcommand{\Esub}[2]{\mathbb E_{#1}\left[#2\right]}

\newcommand{\pp}[1]{\mathbb P\left(#1\right)}
\DeclareMathOperator*{\argmin}{argmin}
\DeclareMathOperator*{\argmax}{argmax}

\newtheorem{theorem}{\bf{Theorem}}
\newtheorem{definition}{\bf{Definition}}
\newtheorem{lemma}{\bf{Lemma}}
\newtheorem{corollary}{\bf{Corollary}}

\newtheorem{remark}{\bf{Remark}}
\newtheorem{example}{\bf{Example}}
\newtheorem{assumption}{\bf{Assumption}}

\begin{document}

\title{Fast Rate Information-theoretic Bounds on Generalization Errors\let\thefootnote\relax\footnotetext{*This work is an extended version of the preliminary work \cite{wu2022fast} appeared in ITW2022 conference. This paper extends the previous results on the justification of the tightness of the generalization error whereas we derive the lower bound for the generalization error that matches the fast rate upper bound. Furthermore, we present instances where the $(\eta,c)$-central condition is not satisfied, supported by several analytical examples. The fast rate results are confirmed with more examples, such as regression and classification problems analytically or numerically, showing the effectiveness of the proposed bounds.}}


\author{Xuetong Wu, 
\thanks{This work is an extended version of the preliminary work \cite{wu2022fast} presented in ITW2022 conference.  The work of Jingge Zhu was supported in part by
the Australian Research Council under Project DE210101497 and Project
DP230101493. 

Xuetong Wu, Jonathan H. Manton and Jingge Zhu are with the Department of Electrical and Electronic Engineering, University of Melbourne, VIC, 3010, Australia. E-mail: wfyitf@gmail.com; \{jmanton, jingle.zhu\}@unimelb.edu.au;}~
Jonathan H. Manton,~\IEEEmembership{Fellow,~IEEE,}~ Uwe Aickelin,\thanks{Uwe Aickelin is with the School of Computing and Information Systems, University of Melbourne, VIC, 3010, Australia. E-mail: uwe.aickelin@unimelb.edu.au.}
Jingge Zhu,~\IEEEmembership{Member,~IEEE}
}


\maketitle

\begin{abstract}

The generalization error of a learning algorithm refers to the discrepancy between the loss of a learning algorithm on training data and that on unseen testing data.  Various information-theoretic bounds on the generalization error have been derived in the literature, where the mutual information between the training data and the hypothesis (the output of the learning algorithm) plays an important role. Focusing on the individual sample mutual information bound by Bu et al.~\cite{bu2020tightening}, which itself is a tightened version of the first bound on the topic by Russo et al.~\cite{russo2016controlling} and Xu et al.~\cite{xu2017information}, this paper investigates the tightness of these bounds, in terms of the dependence of their convergence rates on the sample size $n$. It has been recognized that these bounds are in general not tight, readily verified for the exemplary quadratic Gaussian mean estimation problem, where the individual sample mutual information bound scales as $O(\sqrt{1/n})$ while the true generalization error scales as $O(1/n)$. The first contribution of this paper is to show that the same bound can in fact be asymptotically tight if an appropriate assumption is made. In particular, we show that the fast rate can be recovered when the assumption is made on the excess risk instead of the loss function, which was usually done in existing literature. A theoretical justification is given for this choice.  The second contribution of the paper is a new set of generalization error bounds based on the $(\eta, c)$-central condition, a condition relatively easy to verify and has the property that the mutual information term directly determines the convergence rate of the bound. Several analytical and numerical examples are given to show the effectiveness of these bounds.
\end{abstract}


\section{Introduction} \label{sec:intro}
The generalization error of a learning algorithm refers to the discrepancy between the loss of a learning algorithm on training data and that on unseen testing data. An upper bound on this quantity is crucial for assessing the generalization capability of learning algorithms. Conventionally, many bounding techniques are proposed under different conditions and assumptions.  To name a few, \citet{vapnik1971uniform} proposed VC-dimension, which describes the richness of a hypothesis class for generalization ability.  The notion of ``algorithmic stability" was introduced in \cite{kearns1997algorithmic} and \cite{devroye1979distributiona}  (see also \citet{bousquet2002stability}) for bounding the generalization error by examining if a single training sample has a significant effect on the expected loss. Many of the abovementioned bounds are only concerned with the hypothesis space or the algorithm. For example, VC-dimension methods care about the worst-case bound which only depends on the hypothesis space. The stability methods only specify the properties of learning algorithms but do not require additional assumptions on hypothesis space. To fully characterize the intrinsic nature of a learning problem, it is shown in some recent works that the generalization error can be upper bounded using the information-theoretic quantities \cite{xu2017information,russo2016controlling}, and the bound usually takes the following form:      
\begin{align}
    \mathbb{E}_{W\mathcal{S}_n}[\mathcal{E}(W, \mathcal{S}_n)] \leq \sqrt{\frac{c I(W;\mathcal{S}_n)}{n}}, \label{eq:gen-form}
\end{align}
where the expectation is taken w.r.t. the joint distribution of $W$ and $\mathcal{S}_n$ induced by some algorithm $\mathcal{A}$. Here, $\mathcal{E}(w, \mathcal{S}_n)$ denotes the generalization error (properly defined in~(\ref{eq:gen}) in Section \ref{sec:prob}) for a given hypothesis $w$ and data sample $\mathcal{S}_n = (Z_i)_{i=1,\cdots,n}$, and $I(W;\mathcal{S}_n)$ denotes the mutual information between the hypothesis and data sample, and $c$ is some positive constant. In particular, if the loss function is $\sigma$-sub-Gaussian under the distribution $P_W \otimes P_Z$, $c$ equals $2\sigma^2$. By introducing the mutual information, such a bound gives a bound that depends on both the learning algorithm and the data distribution.  It has been shown that for specific problems, the information-theoretic bounds can recover the results obtained by the method of VC dimension~\cite{xu2017information}, algorithmic stability~\cite{raginsky2016information}, and differential privacy~\cite{steinke2020reasoning} under mild conditions. Further, as pointed out by \cite{asadi2018chaining}, the information-theoretic upper bound could be substantially tighter than the traditional bounds if we could exploit specific properties of the learning algorithm.  We  mention that PAC-Bayes bounds are a class of algorithm-dependent bounds first introduced by \citet{mcallester1999some} (see \cite{alquier2024user} for a recent survey). The PAC-Bayes bounds usually contain a Kullback-divergence term between the learning algorithm and a prior distribution, which can be viewed as a computable relaxation of the information-theoretic bounds, with the difference that most of them are high probability bounds. 

A follow-up work by \cite{bu2020tightening} showed a tightened version of the above bound in the form
\begin{align}
        \mathbb{E}_{W\mathcal{S}_n}[\mathcal{E}(W, \mathcal{S}_n)] \leq \frac{1}{n}\sum_{i=1}^n\sqrt{2\sigma^2 I(W;Z_i)}, \label{eq:gen-form-individual}
\end{align}
which is generally tighter than the bound in (\ref{eq:gen-form}). In particular, while the mutual information term in (\ref{eq:gen-form}) can be arbitrarily large or infinite for deterministic algorithms as observed by  \cite{bu2020tightening,hellstrom2020generalization, steinke2020reasoning}, the individual sample mutual information bound in (\ref{eq:gen-form-individual}) can still give non-vacuous bound for such algorithms. We mention that random subset methods \cite{negrea2019information, rodriguez2021random} (see also the work by \cite{zhou2022individually,haghifam2020sharpened}) and bounds based on ``ghost samples" \cite{steinke2020reasoning} are also used to derive tightened bounds.

It has been recognized that the original bounds in (\ref{eq:gen-form}) or (\ref{eq:gen-form-individual}) are, in general, not tight. It can be easily verified that for the simple Gaussian quadratic problem (details in Example \ref{sec:example}), the bound in (\ref{eq:gen-form-individual}) scales as $O(\sqrt{1/n})$ while the true generalization error scales as $O(1/n)$. It begs the problem of whether information-theoretic bounds can be intrinsically tight. Several works are dedicated to this problem, including \cite{hellstrom2021fast}, \cite{hellstrom2022new}, \cite{zhou2023stochastic} and \cite{zhou2024exactly}. We will compare our results with these works in  Section \ref{subsec:related_works}.

In this work, we develop a general framework for the fast rate bounds using the mutual information following this line of works~\cite{van2015fast,grunwald2020fast,grunwald2021pac} and the contributions are listed as follows.
\begin{itemize}
    \item We show that for the Gaussian quadratic problem, the bound in (\ref{eq:gen-form-individual}) can be made asymptotically tight when an appropriate assumption is made, where the sub-Gaussian condition is assumed on the excess risk instead of on the loss function. An intuitive explanation is given for the choice of this assumption, and a matching lower bound is also given for this specific problem.
    \item Inspired by the analysis for the sub-Gaussian case, we propose a new condition, called the $(\eta, c)$-central condition, under which we derive a set of new generalization bounds. Compared with typical mutual information bounds, the novel bounds have a cleaner presentation and wider applicability. Under the assumption that the excess risk satisfies the $(\eta,c)$-central condition, the bounds have the nice property that the convergence rate is directly determined by the mutual information term $I(W; Z_i)$. The bounds are shown to be asymptotically tight for various examples under the empirical risk minimization (ERM) algorithm.
    \item Furthermore, our results are extended to regularized ERM algorithms, and intermediate rates could be achieved with the relaxed $(v,c)$-central condition. The fast rate results are confirmed for a few simple regression and classification examples analytically or numerically, showing the effectiveness of the proposed bounds. 
\end{itemize}

\subsection{Related Works}\label{subsec:related_works}
There exist several works in the literature,  including \cite{hellstrom2021fast}, \cite{hellstrom2022new}, \cite{zhou2023stochastic}, and \cite{zhou2024exactly}, that aim to obtain a fast rate using information-theoretic bounds. The results in \cite{hellstrom2021fast,hellstrom2022new} contain bounds in terms of conditional mutual information is bounded uniformly in $n$, then a fast rate of $O(1/n)$ can be achieved. The result in \cite{zhou2023stochastic} combines the stochastic chaining technique with the individual mutual information and derives a bound that is asymptotically tight for the quadratic Gaussian problem. A recent work \cite{zhou2024exactly} gives a PAC-Bayes type bound that is shown to be exactly tight for the quadratic Gaussian problem for any $n$. Different from these works, our results emphasize the fact that even the original bound (without other tightening techniques) of (\ref{eq:gen-form-individual}) can give the correct convergence rate under an appropriate assumption. Furthermore, our derived bounds based on the $(\eta,c)$-central condition give a different perspective on the tightness of information-theoretic bounds for generalization. Our proposed $(\eta,c)$-central condition is inspired by the work of \cite{van2015fast,grunwald2020fast,grunwald2021pac}, and is closely related to other conditions such as the $\eta$-central condition \cite{van2015fast}, the witness condition \cite{grunwald2020fast,grunwald2021pac} and the Bernstein condition \cite{bartlett2006empirical}. We give a detailed comparison of these conditions and related bounds in Section \ref{sec:related}.

\section{Problem formulation} \label{sec:prob}
Consider a dataset $\mathcal{S}_n = \left( z_1,z_2,\cdots,z_n\right)$ where each instance $z_i$ is i.i.d. drawn from some distribution $\mu$, we would like to learn a hypothesis $w$ that exploits the properties of $\mathcal{S}_n$, with the aim of making predictions for previously unseen new data correctly. The choice of $w$ is performed within a set of member functions $\mathcal{W}$ with the possibly randomised algorithm $\mathcal{A}:\mathcal{Z}^n \rightarrow \mathcal{W}$, described by the conditional probability $P_{W|\mathcal S_n}$. We define the corresponding loss function $\ell: \mathcal{W}\times \mathcal{Z} \rightarrow \mathbb{R}$. Particularly if we consider the supervised learning problem, we can write $\mathcal{Z} = \mathcal{X} \times \mathcal{Y}$ and $z_i = (x_i, y_i)$ as a feature-label pair. Then the hypothesis $w: \mathcal{X} \rightarrow \mathcal{Y}$ can be regarded as a predictor for the input sample. We will call $(\mu, \ell, \mathcal{W}, \mathcal{A})$ a learning tuple. In a typical statistical learning problem, one may wish to minimize the \emph{expected} loss function $L_{\mu}(w) = E_{z\sim \mu}[\ell(w,z)]$. However, as the underlying distribution $\mu$ is usually unknown in practice, one may wish to learn $w$ by some learning principle, for example, one typical way is minimizing the empirical risk induced by the dataset $\mathcal{S}_n$, denoted as $w_{\ERM}$, such that 
\begin{equation}
w_{\ERM} = \argmin_{w\in \mathcal{W}}\frac{1}{n}\sum_{i=1}^{n}\ell(w,z_i),
\end{equation}
which will be employed as a predictor for the new data. We point out that many of our results obtained in this paper also hold for more general algorithms other than the ERM algorithm. Here we define $\hat{L}(w, \mathcal{S}_n) = \frac{1}{n}\sum_{i=1}^{n}\ell(w,z_i)$ as the empirical loss.  To assess how this predictor performs on unseen samples, the generalization error is then introduced to evaluate whether a learner suffers from over-fitting (or under-fitting). For any $w \in \mathcal{W}$, we define the generalization error as
\begin{equation}
\mathcal{E}(w, \mathcal{S}_n) := \mathbb{E}_{Z\sim \mu}[\ell(w,Z)] - \frac{1}{n}\sum_{i=1}^{n}\ell(w,z_i).  \label{eq:gen}
\end{equation}
Another important metric, the excess risk, is defined as
\begin{equation}
\mathcal{R}(w) := \mathbb{E}_{Z\sim \mu}[\ell(w,Z)] - \mathbb{E}_{Z\sim \mu}[\ell(w^*,Z)].
\end{equation}
where the optimal hypothesis for the true risk is defined as $w^*$ as 
\begin{equation}
w^{*} = \argmin_{w\in \mathcal{W}}E_{Z \sim \mu}[\ell (w,Z)],
\end{equation}
which is unknown in practice. The excess risk evaluates how well a hypothesis $w$ performs with respect to $w^*$ given the data distribution $\mu$. We also define the corresponding empirical excess risk as
\begin{equation}
\hat{\mathcal{R}}(w, \mathcal{S}_n) := \frac{1}{n}\sum_{i=1}^{n}r(w,z_i),
\end{equation}
where $r(w,z) = \ell(w,z) - \ell(w^*,z)$. In the sequel, we are particularly interested in bounding the expected generalization error $\mathbb{E}_{W\mathcal{S}_n}[\mathcal{E}(W, \mathcal{S}_n)]$ and the excess risk $\mathbb{E}_{W}[\mathcal{R}(W)]$  where the distribution $P_W$ is the marginal distribution of $P_{W\mathcal S_n}$, induced by the distribution of the data $\mathcal S_n$ and the algorithm $P_{W|\mathcal S_n}$.

\section{Tightness of Individual Mutual Information Bound}
\subsection{Existing Bounds}
In this section, we first review some known results on generalization error.  A representative result in the following result of individual mutual information bound by \cite{bu2020tightening}, which is a tightened version of the first results on this topic by \cite{xu2017information} and \cite{russo2016controlling}.
\begin{theorem}[Generalization error of Generic Hypothesis \cite{bu2020tightening}] \label{thm:gen_erm}
Assume that the cumulant generating function of the random variable $\ell(W, Z)-\E{\ell(W,Z)}$  is upper bounded by $\psi(-\lambda)$ in the interval $(b_{-},0)$ and $\psi(\lambda)$ in the interval $(0,b_{+})$ under the product distribution $P_W\otimes\mu$ for some $b_{-}<0$ and $b_{+}>0$ where $P_W$ is induced by the data distribution and the algorithm. Then the expectation of the  generalization error in (\ref{eq:gen}) is upper bounded as
\begin{align}
\Esub{W\mathcal{S}_n}{\mathcal{E} (W, \mathcal{S}_n) }\leq \frac{1}{ n}\sum_{i=1}^{ n}\psi^{*-1}_{-}(I(W;Z_i)), \\
-\Esub{W\mathcal{S}_n}{\mathcal{E} (W, \mathcal{S}_n)}\leq \frac{1}{ n}\sum_{i=1}^{n}\psi^{*-1}_{+}(I(W;Z_i)),
\end{align}
where we define
\begin{align}
\psi^{*-1}_{-}(x) :=\inf_{\lambda\in[0,-b_{-})}\frac{x+\psi(-\lambda)}{\lambda}, \\
\psi^{*-1}_{+}(x) :=\inf_{\lambda\in[0,b_{+})}\frac{x+\psi(\lambda)}{\lambda}.
\end{align} \label{thm:exp_gen_gamma}
\end{theorem}
Under different assumptions, the above theorem can be specialized to cases when the loss function is sub-exponential, sub-Gamma, or sub-Gaussian random variables by identifying the CGF bounding function $\psi(\lambda)$. We give the example when the loss function is sub-Gaussian in the following.  More discussions on the CGF bounding function can be found in \cite{jiao2017dependence,bu2020tightening}.

\begin{example}[Sub-Gaussian bound]
We say $X$ is a $\sigma$-sub-Gaussian random variable with variance parameter $\sigma$ if it holds
$\log \mathbb{E}\left[e^{\lambda (X-\mathbb{E}[X])}\right] \leq \frac{\sigma^{2} \lambda^{2}}{2}$ for all $\lambda \in \mathbb{R}$.
Suppose that $\ell({W}, {Z})$ is $\sigma$-sub-Gaussian under the distribution $P_{W} \otimes \mu$ where $P_W$ is the marginal induced the algorithm $\mathcal{A}$ and data distribution $\mu$, then
\begin{align}
    \mathbb{E}_{W\mathcal{S}_n} \left[\mathcal{E}(W, \mathcal{S}_n)\right]  \leq \frac{1}{n} \sum_{i=1}^{n} \sqrt{2 \sigma^{2} I\left(W ; Z_{i}\right)}. \label{eq:bu_result}
\end{align}
\end{example}


If the mutual information term $I(W;Z_i)$ scales as $O(1/n)$ (which is the case for the quadratic Gaussian problem in Example \ref{sec:example}) and if the sub-Gaussian parameter $\sigma^2$ is a constant, then the bound in (\ref{eq:bu_result}) scales as $O(\sqrt{1/n})$ because of the square root. Therefore it is usually recognized that the \emph{square root} sign prevents us from the fast rate, exemplified in the following simple quadratic Gaussian mean estimation problem.
\begin{example}\label{sec:example}
Let $\ell(w,z_i) = (w-z_i)^2$, each sample is drawn from some Gaussian distribution, $Z_i \sim \mathcal{N}(\mu, \sigma_{N}^2)$. We consider the ERM algorithm that gives,
\begin{align*}
 W_{\ERM} = \frac{1}{n} \sum_{i=1}^{n} Z_i \sim \mathcal{N}(\mu, \frac{\sigma_{N}^2}{n}).
\end{align*}
The true generalization error can be calculated to be
 $  \Esub{W\mathcal{S}_n}{\mathcal{E}(W_\ERM, \mathcal{S}_n)} = \frac{2\sigma_{N}^2}{n}$.
To evaluate the upper bound in~Theorem~\ref{thm:gen_erm} for this example, we notice that for any $i$, $\ell(W,Z_i) \sim \frac{n+1}{n}\sigma_{N}^2 \chi_{1}^{2}$ where $\chi^2_1$ denotes the chi-squared distribution with 1 degree of freedom. Hence, the cumulant generating function can be calculated as,
\begin{align*}
\log \mathbb{E}_{P_W\otimes \mu}&\left[e^{\eta(\ell(W,Z)-\mathbb{E}[\ell(W,Z)])}\right] \\
& = - \sigma_{W}^2 \eta -\frac{1}{2} \log \left(1-2\sigma_{W}^2 \eta \right),
\end{align*}
where $\eta \leq \frac{1}{2\sigma^2_W}$ and $\sigma^2_W = \frac{n+1}{n}\sigma_{N}^2$ to simplify the notation. In this case, it can be proved that,
\begin{align}
    - \sigma_{W}^2 \eta -\frac{1}{2} \log \left(1-2\sigma_{W}^2 \eta \right) \leq \sigma_W^4\eta^2. \label{eq:CGF}
\end{align}
We can also calculate the mutual information as 
  $I(W;Z_i) = \frac{1}{2}\log\frac{n}{n-1}$.
With the upper bound on the CGF in (\ref{eq:CGF}), the bound becomes
\begin{align}
    \mathbb{E}_{W\mathcal{S}_n} \left[\mathcal{E}(W, \mathcal{S}_n)\right] & \leq \frac{\sigma_{N}^2}{n} \sum_{i=1}^{n} \sqrt{2\frac{(n+1)^2}{n^2}\log\frac{n}{n-1}} \nonumber \\
    &\leq \sigma^2_N\sqrt{\frac{2(n+1)^2}{(n-1)^3}},
\end{align}
which will be of the order $O(\frac{1}{\sqrt{n}})$ as $n$ goes to infinity, which leads to a slow convergence rate as the true generalization error is $O(\frac{1}{n})$. The detailed calculations can be found in Section IV of~\cite{bu2020tightening}. For completeness and to facilitate further comparison with fast-rate bounds, we also include a full derivation in Appendix~\ref{apd:example2}.
\end{example}
\subsection{Restoring the Fast Rate}
In this section, we show that, in fact, the same bound can be used to derive the correct (fast) convergence rate of $O(1/n)$, with a small yet important change on the assumption. Intuitively speaking, to achieve a fast rate bound for both the generalization error and the excess risk in expectation, the output hypothesis of the learning algorithm must be ``good" enough compared to the optimal hypothesis $w^*$. Here we encode the notion of goodness in terms of the cumulant generating function by controlling the gap between $\ell(w, Z)$ and $\ell(w^*, Z)$. To facilitate such an idea, we make the sub-Gaussian assumption w.r.t. the excess risk and bound the generalization error as follows.
\begin{theorem}\label{thm:sub-Gaussian}
Suppose that $r({W}, {Z})$ is $\sigma$-sub-Gaussian under distribution $P_{W} \otimes \mu$, then
\begin{align}
    \mathbb{E}_{W\mathcal{S}_n} \left[\mathcal{E}(W, \mathcal{S}_n)\right] \leq \frac{1}{n} \sum_{i=1}^{n} \sqrt{2 \sigma^{2} I\left(W ; Z_{i}\right)}. \label{eq:our_result}
\end{align}
Furthermore, the excess risk can be bounded by,
\begin{align}
    \mathbb{E}_{W} \left[\mathcal{R}(W)\right]  \leq &\mathbb{E}_{W\mathcal{S}_n} \left[\hat{\mathcal{R}}(W, \mathcal{S}_n)\right] \nonumber \\
    & + \frac{1}{n} \sum_{i=1}^{n} \sqrt{2 \sigma^{2} I\left(W ; Z_{i}\right)} . \label{eq:our_result_excess}
\end{align}
\end{theorem}
The proof of this result is given in Appendix \ref{proof:sub-Gaussian}. Notice that the expression in (\ref{eq:our_result}) is identical to that in (\ref{eq:bu_result}), with the only difference that the sub-Gaussian condition is assumed on $r(W,Z)$ for the former and on $\ell(W,Z)$ for the latter. Now we evaluate the bound in~Theorem~\ref{thm:sub-Gaussian} for the Gaussian example.
\begin{example}[Continuing from Example~\ref{sec:example}] \label{example:sub-Gaussian-2}
Consider the settings in Example~\ref{sec:example}. First, we note that the expected risk minimizer $w^*$ is calculated as $\mu$. Then we have $r(w,z_i) = (w - z_i)^2 - (\mu - z_i)^2$. The expected excess risk can be calculated as $\mathbb{E}_{W}[\mathcal{R}(W)] = \frac{\sigma_{N}^2}{n}$. We can then calculate the cumulant generating function as,
\begin{align*}
\log \mathbb{E}_{P_W\otimes \mu}\left[e^{\eta(r(W,Z)-\mathbb{E}[r(W,Z)])}\right] \leq \frac{4\eta^2\sigma_N^4}{n},
\end{align*}
for any $\eta \in \mathbb{R}$ and any
$n >  \max\left\{\frac{(4\eta^2\sigma^4_N+\eta\sigma^2_N)(2\eta^2\sigma^4_N+\eta\sigma^2_N)}{\eta^2\sigma^4_N}, 4\eta^2\sigma_N^4 + 2 \eta\sigma_N^2\right\} $
where the detailed calculations can be found in Appendix~\ref{apd:example3}. Hence $r(W,Z)$ is $\sqrt{\frac{8\sigma_N^4}{n}}$-sub-Gaussian under the distribution $P_{W} \otimes \mu$. Then the bound becomes,
\begin{align*}
    \mathbb{E}_{W\mathcal{S}_n} \left[\mathcal{E}(W, \mathcal{S}_n)\right] &\leq \frac{\sigma_N^2}{n} \sum_{i=1}^{n} \sqrt{\frac{8}{n}\log\frac{n}{n-1}} \\
    & \leq \frac{2\sqrt{2}\sigma^2_N}{n-1},
\end{align*}
which is $O(1/n)$, yielding a fast rate characterization.
\end{example}

Unlike prior results that assume the loss function $\ell(W, Z)$ is $\sigma$-sub-Gaussian, we instead assume that the \emph{excess risk} $r(W, Z)$ is $\sigma$-sub-Gaussian. Although the bound in~(\ref{eq:our_result}) takes the same form as~(\ref{eq:bu_result}), the key distinction is that in our setting, $\sigma$ (taking the form $\sqrt{8\sigma_N^4/n}$) varies with the sample size $n$ and in fact vanishes as $n$ increases, leading to a tighter asymptotic bound. This behavior does not arise under the previous assumption, as illustrated in Example~\ref{sec:example}. Moreover, the excess risk can be more directly controlled, as shown in~(\ref{eq:our_result_excess}). While the shift from bounding the loss to bounding the excess risk may initially appear unconventional, we provide further intuition and theoretical justification for this choice in the following subsection.


\subsection{Tightness and Justification for the Quadratic Gaussian Problem} \label{subsec:justification}
In the following, we examine the tightness of the bound and show why $r(W, Z)$ is a more sensible choice, specifically for the quadratic Gaussian problem. To this end, we recall that the main technical tool for deriving the bound is the variational representation of the KL divergence. Specifically, let $X$ be a random variable with alphabet $\mathcal{X}$ and let $P, Q$ be two probability density functions. The KL Divergence admits the following dual representation \cite{donsker1975asymptotic}:
\begin{align}
  D(P \| Q)=\sup _{f: \mathcal{X} \rightarrow \mathbb{R}} \mathbb{E}_P[f(X)]-\log \left(\mathbb{E}_Q\left[e^{f(X)}\right]\right), \label{eq:donsker}
\end{align}
and the tightness of the bound hinges on the choice of the function $f$ in (\ref{eq:donsker}). It is well known that under mild conditions \cite{donsker1975asymptotic}, the optimal function for the Donsker-Varadhan representation is achieved by $f'(dP_{W|Z_i}/dP_W$) where $f(t) = t\log t$. For the quadratic Gaussian problem, we now calculate this optimizer explicitly and show that the choice of $r(W, Z_i)$ is, in fact, the right choice.

To this end, we firstly calculate the densities of $P_W$ and $P_{W|Z_i}$ as: $dP_W = \frac{\sqrt{n}}{\sqrt{2\pi\sigma_N^2}} \exp(\frac{-(W - \mu)^2 n }{2\sigma_N^2})$, $dP_{W|Z_i} = \frac{n}{\sqrt{2\pi \sigma_N^2(n-1)}} \exp(-\frac{(W - \frac{n-1}{n}\mu - \frac{1}{n}Z_i)^2n^2}{2\sigma_N^2(n-1)}).$ Then we can calculate the optimizer as:
\begin{align*}
    &f'(dP_{W|Z_i}/dP_W) = \log \frac{dP_{W|Z_i}}{dP_W} + 1 \\
    &= - \frac{(w-z_i)^2 - (\mu - z_i)^2}{2\sigma_N^2} - \frac{(w-z_i)^2}{2\sigma_N^2(n-1)} \\
     & \quad + 1 + \frac{1}{2}\log\frac{n}{n-1} 
\end{align*}
for fixed $w$ and $z_i$. The above function can be written as:
\begin{align*}
    f'(dP_{W|Z_i}/dP_W) = &-\frac{r(w,z_i)}{2\sigma_N^2} - \frac{\ell(w,z_i)}{2\sigma_N^2(n-1)} \\
    & + \frac{1}{2}\log\frac{n}{n-1} + 1.
\end{align*}
The unexpected excess risk $r(w,z_i)$ clearly appears in the optimizer with some scaling factor and shifting constant, which, however, will not affect the convergence. To rigorously show this, we state the following result.
\begin{lemma}\label{lemma:tightness}
Consider the quadratic Gaussian problem in Example \ref{sec:example} and the Donsker-Varadhan representation in (\ref{eq:donsker}). Let $\eta = \frac{1}{2\sigma^2_N}$, the function $R =-\eta r(w,z_i)$ satisfies the following inequality:
\begin{align*}
    \frac{n-1}{n} I(W;Z_i) &\leq \mathbb{E}_{WZ_i}[R] - \log \mathbb{E}_{P_W\otimes \mu}[e^{R}] \\
    &\leq I(W;Z_i),
\end{align*}
\end{lemma}
The proof of the above Lemma is given in Appendix \ref{proof:lemma_tightness}.  The above lemma demonstrates that the variational representation is essentially tight for the quadratic Gaussian problem. As in this case, the RHS of (\ref{eq:donsker}) is essentially equal to the divergence $D(P||Q)$ (which is the mutual information $I(W; Z_i)$ for the quadratic Gaussian problem) on the LHS. On the other hand, it can be checked straightforwardly that the loss function $\ell(W, Z)$ does not admit a tight approximation for $I(W; Z_i)$.

For the quadratic Gaussian problem, we can, in fact, show that mutual information also gives a lower bound for both the excess risk and generalization error for the Gaussian mean estimation problem.
\begin{lemma}[Matching Lower Bound]\label{thm:lower_bounds}
Consider the quadratic Gaussian mean estimation problem with the ERM algorithm. With a large $n$, we have: 
\begin{align*}
    &\mathbb{E}_{P_W \otimes \mu}[r(W,Z)] \geq  \frac{2\sigma^2_N}{n}\sum_{i=1}^{n}I(W;Z_i) \\
    & + \mathbb{E}_{W\mathcal{S}_n}[\hat{\mathcal{R}}(W,\mathcal{S}_n)] - \frac{1}{n-1} \mathbb{E}_{W\mathcal{S}_n}[\mathcal{E}(W,\mathcal{S}_n)].
\end{align*}
For the generalization error, we have:
\begin{align*}
    \mathbb{E}_{W \mathcal{S}_n}[\mathcal{E}(W,\mathcal{S}_n)] \geq 2\sigma_N^2 \frac{n-1}{n^2}\sum_{i=1}^{n} I(W;Z_i).
\end{align*}
\end{lemma}
The proof of this result is given in Appendix \ref{proof:lowerbound}. From the above results, we observe that the individual sample mutual information appears in both the upper and lower bounds. For the generalization error, the upper and lower bounds are matched regarding the convergence rate with different leading constants. For the excess risk in the Gaussian mean example, the upper bound is tight since the empirical excess risk and generalization error are both of $O(\frac{1}{n})$. 

\section{Novel Fast Rate Results}
\subsection{New Bounds with $(\eta,c)$-Central Condition}

As observed in the results from the previous section, while our new bound is tight, it still contains a square root term in the bound. In fact, many information-theoretic bounds for generalization error \cite{xu2017information, raginsky2016information, bu2020tightening, zhou2022individually, steinke2020reasoning, russo2016controlling} contain the square root with the sub-Gaussian assumption, which is often seen as the obstacle to achieving fast rate results. As the first result in this section, we provide an alternative bound based on the sub-Gaussian assumption that does not contain the square root. 

\begin{theorem}[Fast Rate with Sub-Gaussian Condition]\label{thm:sub-Gaussianv2}
Assume that $r(W, Z)$ is $\sigma$-sub-Gaussian under the distribution $P_W \otimes \mu$. Then it holds that
\begin{align}
     \mathbb{E}_{W\mathcal{S}_n} \left[\mathcal{E}(W, \mathcal{S}_n)\right] \leq & \frac{1-a_\eta}{a_\eta} \mathbb{E}_{W\mathcal{S}_n}[\hat{\mathcal{R}(W,\mathcal{S}_n)}] \nonumber \\
     & + \frac{1}{n\eta a_\eta} \sum_{i=1}^{n}  I\left(W ; Z_{i}\right). \label{eq:sub-Gaussian}
\end{align}
for any $ 0 < \eta < \frac{2\mathbb{E}_{P_W \otimes \mu}[r(W,Z_i)]}{\sigma^2}$ and $a_\eta = 1-  \frac{\eta\sigma^2}{2\mathbb{E}_{P_W \otimes \mu}[r(W,Z_i)]}$. Furthermore, the expected excess risk is bounded by:
\begin{align*}
 \mathbb{E}_{W}[\mathcal{R}(W)] \leq & \frac{1}{a_\eta} \mathbb{E}_{W\mathcal{S}_n}[\hat{\mathcal{R}(W,\mathcal{S}_n)}] \\
 & + \frac{1}{n\eta a_\eta} \sum_{i=1}^{n}  I\left(W ; Z_{i}\right).
\end{align*}
\end{theorem}
The proof of this result is given in Appendix \ref{proof:sub-Gaussianv2}. We continue to examine the bound in Theorem~\ref{thm:sub-Gaussianv2} with the Gaussian mean estimation. 
\begin{example}\label{eg:subv2}
Since the expected excess risk can be calculated as $\mathbb{E}_{W}[\mathcal{R}(W)] = \frac{\sigma_N^2}{n}$, and $r(W,Z)$ is $\sqrt{\frac{8\sigma_N^4}{n}}$-sub-Gaussian, then we require that $0 <\eta < \frac{1}{2\sqrt{2}\sigma_N^2}$, which is independent of the sample size. For simplicity, we can apply Theorem \ref{thm:sub-Gaussianv2} to the quadratic Gaussian mean estimation problem with the choice $\eta = \frac{1}{4\sigma_N^2}$ and $a_\eta=\frac{1}{2}$. For any $n$ that satisfies the condition in Example~\ref{example:sub-Gaussian-2}, we have the generalization error bound,
\begin{align*}
   \frac{1-a_\eta}{a_\eta}&\Esub{W\mathcal{S}_n}{\hat{\mathcal{R}}(W_\ERM,\mathcal{S}_n)} +  \frac{1}{\eta a_{\eta} n}\sum_{i=1}^{n}I(W;Z_i) \\
   &\leq  \frac{3\sigma_N^2}{n},
\end{align*}
where the empirical excess risk $\Esub{W\mathcal{S}_n}{\hat{\mathcal{R}}(W_\ERM,\mathcal{S}_n)}$ is calculated as $-\frac{\sigma^2_N}{n}$ and the bound has the rate of $O(1/n)$. 
\end{example}

Notice that in the above example, both $\eta$ and $\alpha_\eta$ depend on the expected excess risk $\mathbb{E}_{P_W \otimes \mu}[r(W,Z)]$ and $\sigma^2$, which potentially depend on $n$ as well.  In this sense, this result is still not very satisfying, as this dependence makes it hard to determine the actual convergence rate directly from the bounds. To this end, we propose a different type of bound to alleviate this drawback. Ideally, we aim for a bound that does not contain extra quantities that depend on $n$, the key to such a bound is the so-called expected ($\eta,c$)-central condition (or we simply say ($\eta,c$)-central condition for short), inspired by the works \cite{van2015fast,mehta2017fast,grunwald2020fast,grunwald2021pac}. We will first define the $(\eta,c)$-central condition for a non-negative random variable and use this notation to bound the generalization error with the loss function and excess risk.

\begin{definition}[$(\eta,c)$-central condition]
Let ${\eta}>0$ and $0 < c \leq 1$ be two constants. We say that a random variable $X$ endowed with a probability measure $P$ satisfies the $(\eta,c)$-central condition if the following inequality holds:
\begin{align}
\log \mathbb{E}_{P}& \left[e^{-{\eta}X}\right]  \leq   -c\eta  \mathbb{E}_{P}\left[X\right]. \label{eq:eta_c_loss} 
\end{align} 
given that $0 < E_P[X] < \infty$.
\end{definition}
\begin{remark}
Comparing the $(\eta,c)$-central condition with the $\sigma$-sub-Gaussian condition defined as
\begin{align*}
    \log \mathbb{E}_{P}[e^{-\eta X}] \leq -\eta \mathbb{E}_{P}[X] + \frac{\eta^2\sigma^2}{2}, \forall \eta \in \mathbb{R},
\end{align*}
One difference between these two conditions is that the $(\eta, c)$-central condition is required to hold just for a single specific $\eta$, whereas the sub-Gaussian condition is required for all $\eta$. Another difference is on the bounding terms for the CGF where the variance proxy term $\frac{\eta^2\sigma^2}{2}$ is replaced by the term $(1-c)\eta \Esub{P}{X}$ for a positive $\eta$. This also shows that if we want to make the sub-Gaussian condition and the $(\eta,c)$-central condition equivalent, the variance proxy $\sigma^2$ should have the same scaling law as the term $\Esub{P}{X}$. This is consistent with the results in Example \ref{example:sub-Gaussian-2}, where both $\E{\mathcal R(W)}$ and the variance proxy of $r(W,Z)$ scales as $O(1/n)$. We furthermore compare $(\eta,c)$-central condition with other related conditions in the literature in Section \ref{sec:related}.
\end{remark}

Then we make the following assumption that the unexpected excess risk satisfies the $(\eta,c)$-central condition.
\begin{assumption}\label{assump:eta_c_r}
We assume the unexpected excess risk $r(W,Z_i)$ satisfies the $\left(\eta, c \right)$-central condition for any $i\in [n]$, some constants $\eta > 0$ and $0 < c \leq 1$ under the product distribution $P_W\otimes\mu$, e.g.,
\begin{align}
\log \mathbb{E}_{P_W\otimes \mu}& \left[e^{-{\eta}\left(\ell(W,Z)-\ell(w^*,Z)\right)}\right]  \leq \nonumber \\
& -c\eta  \mathbb{E}_{P_W\otimes \mu}\left[\ell(W,Z) - \ell(w^*,Z)\right]. \label{eq:eta_c}
\end{align} 
\end{assumption}

The above definition is inspired by the $\eta$-central condition \cite[(5)]{van2015fast}, which in our notation takes the form
\begin{align}
\log \mathbb{E}_{P_W\otimes \mu}& \left[e^{-{\eta}\left(\ell(W,Z)-\ell(w^*,Z)\right)}\right]  \leq   0  \label{eq:eta},
\end{align} 
which is the special case of $(\eta,c)$-central condition with $c=0$. Compared to the $\eta$-central condition, the RHS of~(\ref{eq:eta_c}) is negative and has a tighter control than (\ref{eq:eta}) of the tail behaviour for some $c> 0$. A similar condition is given in \cite[(4)]{van2015fast}, which in our notation  takes the form
\begin{align}
\log \mathbb{E}_{\mu}& \left[e^{-{\eta}\left(\ell(w,Z)-\ell(w^*,Z)\right)}\right]  \leq 0, \forall w \in \mathcal{W}.
\end{align}
This is a stronger condition as it is required to hold all $w \in \mathcal{W}$ instead of in expectation.


\begin{remark}
Using the Markov inequality, the $(\eta,c)$-central condition on $r(W,Z)$ implies that for any $t$, we have
\begin{align*}
    \pp{r(W,Z)\leq t}&=\pp{-r(W,Z)\geq -t} \\
    &=\pp{e^{-\eta r(W,Z)}\geq e^{-\eta t}}\\
    &\leq e^{\eta t}\E{e^{-\eta r(W,Z)}} \\
    &\leq e^{-\eta(c\E{r(W,Z)}-t)}.
\end{align*}
This shows for $t< c\E{r(W,Z)}$, the probability that  $\ell(W,Z)\leq \ell(w^*,Z)+t$ is exponentially small.
\end{remark}



With the definitions in place, we derive the fast rate bounds under the $(\eta, c)$-central condition as follows. 

\begin{theorem}[Fast Rate with $(\eta, c)$-central condition]\label{thm:eta-c}
Suppose Assumption~\ref{assump:eta_c_r} is satisfied, then it holds that:
\begin{align*}
     \mathbb{E}_{W\mathcal{S}_n}[\mathcal{E}(W,\mathcal{S}_n)] \leq & \frac{1-c}{c} \mathbb{E}_{P_{W\mathcal{S}_n}}[\hat{\mathcal{R}}\left(W, \mathcal{S}_{n} \right)] \\
     & + \frac{1}{c\eta n} \sum_{i=1}^{n} I(W;Z_i).
\end{align*}
\noindent Furthermore, the excess risk is bounded by,
 \begin{align*}
     \mathbb{E}_{W}[\mathcal{R}(W)] \leq & \frac{1}{c} \mathbb{E}_{P_{W\mathcal{S}_n}}[\hat{\mathcal{R}}\left(W, \mathcal{S}_{n} \right)] \\
     & + \frac{1}{c\eta n} \sum_{i=1}^{n} I(W;Z_i).
 \end{align*}
\end{theorem}

The proof of this result is given in Appendix \ref{proof:thm_eta-c}. Such a bound retrieves a similar form with \cite[Eq.~(3)]{grunwald2021pac}, which consists of the empirical excess risk and mutual information terms. For some algorithms, such as ERM, the first term is non-positive, allowing the bound to be tightened and expressed solely in terms of mutual information to achieve fast rates. Notice that different from the bound in Theorem~\ref{thm:sub-Gaussianv2}, the bound in Theorem~\ref{thm:eta-c} contains constants $c$ and $\eta$ that in general do not depend on the sample size $n$ (as we will verify for several examples later in the paper.) Therefore, this bound has the nice property that the convergence rate is solely determined by the mutual information term $I(W; Z_i)$.

In the following, we analytically examine our bounds in Gaussian mean estimation, and we also empirically verify our bounds with a logistic regression problem in Appendix~\ref{sec:logistic}. 
\begin{example}
We can verify that the $(\eta, c)$-central condition is satisfied for the quadratic Gaussian mean estimation problem. It can be checked that for $n >  \max\left\{\frac{(4\eta^2\sigma^4_N-\eta\sigma^2_N)(2\eta^2\sigma^4_N-\eta\sigma^2_N)}{\eta^2\sigma^4_N}, 4\eta^2\sigma_N^4 - 2 \eta\sigma_N^2\right\} $, it holds that
\begin{align*}
\log \mathbb{E}_{P_W\otimes \mu}\left[e^{-\eta r(W,Z)}\right] \leq \frac{4 \eta^2\sigma_N^4 - \eta \sigma_N^2}{n} \leq -c\eta \frac{\sigma_N^2}{n}.
\end{align*}
From the above inequality, this learning problem satisfies the $(\eta, c)$-central condition for any $0 < \eta < \frac{1}{4\sigma_N^2}$ and any $c \leq 1- 4\eta\sigma_N^2$, which is independent of the sample size and thus does not affect the convergence rate. Similarly, take $\eta = \frac{1}{8\sigma_N^2}$ and $c = \frac{1}{2}$, the bound becomes
\begin{align*}
    \frac{1-c}{c} \mathbb{E}_{P_{W\mathcal{S}_n}}[\hat{\mathcal{R}}\left(W, \mathcal{S}_{n} \right)] + \frac{1}{c\eta n} \sum_{i=1}^{n} I(W;Z_i) \leq \frac{7\sigma_N^2}{n},
\end{align*}
which coincides with the bound in~Example~\ref{eg:subv2}.
\end{example}

It is natural to consider whether we can apply the ($\eta,c$)-central condition to the loss function, and whether that leads to a good bound. We present the following theorem regarding the generalization error when the loss function satisfies the $(\eta,c)$-central condition, and evaluate the resulting bound for the Gaussian mean estimation problem.

\begin{theorem}[Generalization Error Bounds with $(\eta, c)$-central condition w.r.t. the loss function]\label{thm:eta-c-loss}
Assume the loss function $\ell(w,z_i)$ satisfies the  $\left(\eta, c \right)$-central condition for any $i\in [n]$, some constants $\eta > 0$ and $0 < c \leq 1$ under the data distribution $\mu$ and algorithm $\mathcal{A}$. Then it holds that:
\begin{align}
     \mathbb{E}_{W\mathcal{S}_n}[\mathcal{E}(W,\mathcal{S}_n)] \leq & \frac{1-c}{c}\mathbb{E}_{W\mathcal{S}_n}[\hat{L}(W,\mathcal{S}_n)] \\
     & + \frac{\sum_{i=1}^{n}I(W;Z_i)}{c\eta n}. \label{eq:loss-bound}
\end{align}
\end{theorem}


\begin{example}
Concretely, let us again examine the Gaussian mean estimation problem. To verify the ($\eta,c$)-central condition w.r.t. the loss function, we can calculate the CGF under the distribution $P_W \otimes \mu$ with the $(\eta, c)$-central condition as:
\begin{align}
    \log \mathbb{E}_{P_W\otimes \mu}[e^{-\eta \ell(W,Z)}] &= -\frac{1}{2}\log(1+2\eta \frac{n+1}{n} \sigma_N^2) \nonumber \\
    &\leq -2c\eta \frac{n+1}{n}\sigma_N^2,
\end{align}
For arbitrary choice of $\eta$ and any $n \geq 1$, we can select $c = \frac{\log(1+4\eta \sigma^2_N)}{4\eta \sigma^2_N}$ as the function $\frac{\log(1+x)}{x}$ is non-increasing for all $x > 0$. The generalization error can be upper bounded by,
\begin{align}
     \mathbb{E}_{W\mathcal{S}_n}[\mathcal{E}(W,\mathcal{S}_n)] \leq & \frac{1-c}{c}\frac{n-1}{n}\sigma^2_N + \frac{1}{2c\eta(n-1)}. \label{eq:eta-c-gaussian-loss}
\end{align}
which converges to the (non-zero) constant $\frac{1-c}{c}\sigma^2_N$. Hence, the bound is not tight in this case.
\end{example}

\begin{remark}
As can be seen from the above example, the R.H.S. of (\ref{eq:eta-c-gaussian-loss}) does not converge to zero with $n$ increasing when $c$ and $\eta$ are chosen to be independent of the sample size $n$, which does not match the true generalization error. We can draw a comparable insight from this example as illustrated in the previous section: obtaining the fast rate result requires us to make appropriate assumptions on the unexpected excess risk $r(W, Z)$ while making assumptions under the loss function itself is not sufficient to a tight bound.
\end{remark}

\subsection{Connection with Other Conditions} \label{sec:related}
Fast rate conditions are widely investigated under different learning frameworks and conditions \cite{van2015fast, mehta2017fast, koren2015fast, mhammedi2019pac, grunwald2020fast, zhu2020semi, grunwald2021pac}. As the most relevant work, our bound is similar to that found in \cite{grunwald2021pac} which applies conditional mutual information \cite{steinke2020reasoning}, but their results do not hold for unbounded losses and specifically do not hold for sub-Gaussian losses. Our result applies to more general algorithms with a weaker assumption since we only require the proposed conditions to hold in expectation w.r.t. $P_W$, instead of for all $w \in\mathcal{W}$. Our results also have the benefit of allowing the convergence factors to be further improved by using different metrics and data-processing techniques, see \cite{jiao2017dependence, hafez2020conditioning, zhou2022individually} for examples.  Now it is instructive to compare the different assumptions used in the related works. We point out that the $(\eta,c)$-central condition is indeed the key assumption for generalizing the result of Theorem~\ref{thm:sub-Gaussianv2}, which also coincides with some well-known conditions that lead to a fast rate. We firstly show that the Bernstein condition \cite{bartlett2006empirical,bartlett2006convexity, hanneke2016refined,mhammedi2019pac} implies the $(\eta,c)$-central condition for certain $\eta$ and $c$ in the following corollary. 
\begin{corollary}\label{coro:berstein}
Let $\beta \in [0,1]$ and $B \geq 1$. For a learning tuple $(\mu, \ell, \mathcal{W}, \mathcal{A})$,  we say that the \textbf{Bernstein condition} holds if the following inequality holds for the optimal hypothesis $w^*$:
\begin{align*}
       {\mathbb{E}}_{P_W \otimes \mu}&\left[\left(\ell\left(W, Z^{\prime}\right)-\ell\left(w^{*} , Z^{\prime}\right)\right)^{2}\right] \\
       & \leq B \left({\mathbb{E}}_{P_W \otimes \mu}\left[\ell\left(W, Z^{\prime}\right)- \ell\left(w^{*} ; Z^{\prime}\right)\right]\right)^{\beta}.
\end{align*}
If the above condition holds with $\beta = 1$ and $r(w,z_i)$ is lower bounded by $-b$ with some $b > 0$ for all $w$ and $z_i$, the learning tuple also satisfies $(\min(\frac{1}{b}, \frac{1}{2B(e-2)}), \frac{1}{2})$-central condition.
\end{corollary}

The proof of the above result is given in Appendix \ref{proof:corollary_bernstein}. The Bernstein condition is usually recognized as a characterization of the ``easiness" of the learning problem under various $\beta$ where $\beta = 1$ corresponds to the ``easiest" learning case. For bounded loss functions, the Bernstein condition will automatically hold with $\beta = 0$.  Different from the standard setting, we only require that the learned (randomized) hypothesis $W$ satisfy the inequality in expectation. This is a weaker requirement: the condition does not need to be satisfied for all $w\in\mathcal W$ but only needs to be satisfied on average.

The second condition is the central condition with the witness condition \cite{van2015fast,grunwald2020fast}, which also implies the $(\eta,c)$-central condition. We say $(\mu, \ell, \mathcal{W}, \mathcal{A})$ satisfies the $\eta$-central condition \cite{van2015fast,grunwald2020fast} if for the optimal hypothesis $w^*$, the following inequality holds,
\begin{align*}
\mathbb{E}_{P_W\otimes \mu}\left[e^{-{\eta}\left(\ell(W,Z)-\ell(w^*,Z)\right)}\right] \leq 1. 
\end{align*}
We also say the learning tuple  $(\mu, \ell, \mathcal{W}, \mathcal{A})$ satisfies the  $(u, c)$-witness condition \cite{grunwald2020fast} if for constants $u > 0$ and $c \in (0,1]$, the following inequality holds.
\begin{align*}
     &\mathbb{E}_{P_W\otimes \mu} [\left(\ell(W,Z)-\ell({w^{*}},Z) \right) \cdot   \mathbf{1}_{\left\{\ell(W,Z)-\ell({w^{*}},Z) \leq u \right\}}] \\
     &\geq c \mathbb{E}_{P_W\otimes \mu}\left[\ell(W,Z) -\ell({w^{*}},Z) \right],
\end{align*}
where $ \mathbf{1}_{\{\cdot\}}$ denotes the indicator function.  The standard $\eta$-central condition is a key condition for proving the fast rate \cite{van2015fast,mehta2017fast,grunwald2020fast}. Some examples are exponential concave loss functions (including log-loss) with $\eta = 1$ (see \cite{mehta2017fast,zhu2020semi} for examples) and bounded loss functions with Massart noise condition with different $\eta$ \cite{van2015fast}.  The witness condition~\cite[Def. 12]{grunwald2020fast} is imposed to rule out situations in which learnability simply cannot hold. The intuitive interpretation of this condition is that we exclude bad hypothesis $w$ with negligible probability (but still can contribute to the expected loss), which we will never witness empirically.  They are connected to the $(\eta,c)$-central condition in the following way.
\begin{corollary}\label{coro:central}
If the learning tuple satisfies both $\eta$-central condition and $(u,c)$-witness condition, then the learning tuple also satisfies the $(\eta', \frac{c(1-\eta'/\eta)}{\eta' u +1})$-central condition for any $0< \eta' < \eta$.
\end{corollary}

Furthermore, we can show the following connection between the $(\eta,c)$-central condition to sub-exponential and sub-gamma assumptions. Recall that we say $X$ is a $\left(\nu^2, \alpha\right)$-sub-exponential random variable with parameters $\nu, \alpha>0$ if:
\begin{align}
\log\mathbb{E} \left[ e^{\lambda (X -\mathbb{E}[X])}\right] \leq \frac{\lambda^2 \nu^2}{2} , \quad \forall \lambda:|\lambda|<\frac{1}{\alpha}.  
\end{align}
We say $X$ is a ($\nu^2,\alpha$)-sub-Gamma random variable with variance parameter $\nu^2$ and scale parameter $\alpha$ if:
\begin{align}
\log\mathbb{E} \left[ e^{\lambda (X -\mathbb{E}[X])}\right]\leq \frac{\nu^2 \lambda^2}{2(1- \alpha \lambda)}, \quad  \forall \lambda:  0<\lambda<\frac{1}{\alpha}.
\end{align}


The following result is obtained by combining the sub-exponential/Gamma assumption with the $(\eta,c)$-central condition.

\begin{corollary}\label{coro:subexponential}
If $r(W, Z)$ is ($\nu^2$, $\alpha$)-sub-exponential under the distribution $P_W \otimes \mu$, then the learning tuple satisfies $(\min(\frac{1}{\alpha}, \frac{\nu^2}{\mathbb{E}_{P_W \otimes \mu}[r(W,Z)]}), \frac{1}{2})$-central condition. Correspondingly, the generalization error is upper bounded as
\begin{align}
     \mathbb{E}_{W\mathcal{S}_n}[\mathcal{E}(W,\mathcal{S}_n)] \leq &  \mathbb{E}_{P_{W\mathcal{S}_n}}[\hat{\mathcal{R}}\left(W, \mathcal{S}_{n} \right)] \nonumber \\
     &+ \frac{2}{\eta n} \sum_{i=1}^{n} I(W;Z_i). \label{sub-exponential}
\end{align}
for any $0< \eta \leq \min(\frac{1}{\alpha}, \frac{\nu^2}{\mathbb{E}_{P_W \otimes \mu}[r(W,Z)]})$. If $r(W, Z)$ is ($\nu^2$, $\alpha$)-sub-Gamma under the distribution $P_W \otimes \mu$, then the learning tuple satisfies $(\frac{\mathbb{E}_{P_W\otimes\mu}[r(W,Z)]}{\nu^2+\alpha \mathbb{E}_{P_W\otimes\mu}[r(W,Z)]}, \frac{1}{2})$-central condition.  Correspondingly, the generalization error is upper bounded in the same way as in (\ref{sub-exponential})
for any $0< \eta \leq\frac{\mathbb{E}_{P_W\otimes\mu}[r(W,Z)]}{\nu^2+\alpha \mathbb{E}_{P_W\otimes\mu}[r(W,Z)]}$.
\end{corollary}

The proof of the above result is given in Appendix~\ref{proof:coro_subexponential_subgamma}. For a simple comparison, we present the following corollary, obtained by directly applying the known result Theorem \ref{thm:gen_erm} with the sub-exponential and sub-Gamma assumption. We omit the detail of the proof, as this result is very similar to  Corollaries 2 and 3 in \cite{jiao2017dependence}, with the difference that we consider the individual mutual information version of the bound. 

\begin{corollary}\label{coro:subexponential_old}
    Assume that $\ell(W, Z)$ is ($\nu^2$, $\alpha$)-sub-exponential under the distribution $P_W \otimes \mu$ for some $\nu^2$ and $\alpha > 0$. Then it holds that
\begin{multline}
\mathbb{E}_{W\mathcal{S}_n} \left[\mathcal{E}(W, \mathcal{S}_n)\right] \\
\leq \begin{cases}
 \frac{1}{n}\sum_{i=1}^{n}\sqrt{2\nu^2 I(W;Z_i)}, & \text{if } I(W;Z_i) \leq \frac{\nu^2}{2\alpha^2} \\
 & \text{for all } i \in [n],\\
 \frac{\nu^2}{2\alpha} + \frac{\alpha}{n}\sum_{i=1}^{n} I(W;Z_i), & \text{if } I(W;Z_i) > \frac{\nu^2}{2\alpha^2} \\
 & \text{for all } i \in [n].
 \end{cases}
 \label{eq:subgexponential}
\end{multline}
Assume that $\ell(W, Z)$ is ($\nu^2$, $\alpha$)-sub-Gamma under the distribution $P_W \otimes \mu$ for some $\nu^2$ and $\alpha > 0$. Then it holds that
\begin{align}
     \mathbb{E}_{W\mathcal{S}_n} \left[\mathcal{E}(W, \mathcal{S}_n)\right] \leq &\frac{1}{n}\sum_{i=1}^{n}\sqrt{2\nu^2I(W;Z_i)} \nonumber \\
     & + \alpha I(W;Z_i). \label{eq:subgamma}
\end{align}
\end{corollary}

A detailed evaluation of the practical utility of Corollary \ref{coro:subexponential} and its potential advantages over Corollary \ref{coro:subexponential_old} is left for future work. To end this section, we summarize and outline all  technical conditions in Table~\ref{tab:tech1}. From the table, we can see that our proposed $(\eta,c)$-central condition coincides with many existing works such as \cite{grunwald2020fast} and \cite{grunwald2021pac} with suitable choices of $c$ and $\eta$.  For bounded loss, $\beta = 1$ in the Bernstein condition is equivalent to the central condition with the witness condition,  which implies the $(\eta,c)$-central condition. As an example of unbounded loss functions, the log-loss will satisfy the central and witness conditions under well-specified model \cite{wong1995probability,grunwald2020fast}, which also consequently implies the $(\eta,c)$-central condition. As suggested by Theorem~\ref{thm:sub-Gaussianv2}, Corollary~\ref{coro:subexponential}, the sub-Gaussian, sub-exponential and sub-Gamma conditions can also satisfy the $(\eta,c)$-central condition for different parameters in the assumptions.  

\begin{table*}[!ht]
    \centering
    \caption{A comparison of different conditions in Section \ref{sec:related}}\label{tab:tech1}
    \begin{tabular}{|c|c|}
    \hline
     Condition      &  Key Inequality   \\
     \hline
     $(\eta,c)$-Central Condition & $\log \mathbb{E}\left[e^{-\eta r(W,Z)}\right] \leq  -c \eta \mathbb{E}[r(W,Z)]$ \\
     \hline 
     Bernstein Condition with $\beta = 1$   &  $\log  \mathbb{E}\left[e^{-\eta r(W,Z)}\right] \leq  -\frac{1}{2} \eta \mathbb{E}[r(W,Z)]$ \\
     \hline 
     Central  Condition + Witness Condition    &  $\log  \mathbb{E}\left[e^{-\eta r(W,Z)}\right] \leq  -\frac{1}{c_u}\eta\mathbb{E}[r(W,Z)]$   \\
    \hline 
     Central Condition  &  $\log  \mathbb{E}\left[e^{-\eta r(W,Z)}\right]  \leq 0$      \\
     \hline
     Sub-Gaussian Condition  & $\log \mathbb{E}\left[ e^{ -\eta r(W,Z)} \right] \leq -\eta \mathbb{E}[ r(W,Z)] +\frac{\eta^2 \sigma^2}{2}$  for $\eta\in\mathbb R$   \\
     \hline 
     Sub-exponential Condition & $\log \mathbb{E}\left[ e^{ -\eta r(W,Z)} \right] \leq -\eta \mathbb{E}[ r(W,Z)] +\frac{\eta^2 \nu^2}{2}$  for  $|\eta|\leq 1/\alpha$  \\
     \hline 
     Sub-Gamma Condition & $\log \mathbb{E}\left[ e^{ -\eta r(W,Z)} \right] \leq -\eta \mathbb{E}[ r(W,Z)]  + \frac{\nu^2 \lambda^2}{2(1- \alpha \lambda)}$ for $0\leq \eta\leq 1/\alpha$ \\
     \hline 
    \end{tabular}
\end{table*}

\subsection{Extensions to intermediate rates}
The regularized ERM algorithm, which involves minimizing the empirical risk function and a regularization term, is often used in practice for better statistical and computational properties of the optimization problem. In this section, we further apply the learning bound in Theorem~\ref{thm:eta-c} to the regularized ERM algorithm with the following optimization problem:
\begin{align*}
    w_{\sf{RERM}} = \argmin_{w \in \mathcal{W}} \hat{L}(w,\mathcal{S}_n) + \frac{\lambda}{n}g(w),
\end{align*}
where $g : \mathcal{W} \rightarrow \mathbb{R}$ denotes the regularizer function and $\lambda$ is some coefficient. We define $\hat{\mathcal{R}}_{\textup{reg}}(w,\mathcal{S}_n) = \hat{\mathcal{R}}(w,\mathcal{S}_n) + \frac{\lambda}{n}(g(w) - g(w^*))$, then we have  the following lemma.
\begin{corollary}\label{coro:rerm}
Suppose Assumption~\ref{assump:eta_c_r} hold and also assume $|g(w_1) - g(w_2)| \leq B$ for any $w_1$ and $w_2$ in $\mathcal{W}$ with some $B >0$. Then for the regularized ERM hypothesis $W_{\sf{RERM}}$:
\begin{align*}
     \mathbb{E}_{W}[\mathcal{R}(W_{\sf{RERM}})] \leq & \frac{1}{c} \mathbb{E}_{P_{W\mathcal{S}_n}}[\hat{\mathcal{R}}_{\textup{reg}}\left(W_{\sf{RERM}}, \mathcal{S}_{n} \right)]  \\
     &+\frac{\lambda B}{cn} + \frac{1}{c\eta n} \sum_{i=1}^{n} I(W_{\sf{RERM}};Z_i). 
\end{align*}
\end{corollary}
The proof of this result is given in Appendix \ref{proof:rerm}. As $\hat{\mathcal{R}}_{\textup{reg}}(w,\mathcal{S}_n)$ will be negative for $w_{\sf{RERM}}$, the regularized ERM algorithm can lead to the fast rate if $I(W_{\sf{RERM}};Z_i) \sim O(1/n)$, which coincides with results in \cite{koren2015fast}.

From Theorem~\ref{thm:eta-c}, we can achieve the linear convergence rate $O(1/n)$ if the mutual information between the hypothesis and data example is converging with $O(1/n)$. To further relax the $(\eta,c)$-central condition, we can also derive the intermediate rate with the order of $O(n^{-\alpha})$ for $\alpha \in [\frac{1}{2}, 1]$. Similar to the $v$-central condition, which is a weaker condition of the $\eta$-central condition \cite{van2015fast,grunwald2020fast}, we propose the $(v,c)$-central condition  and derive the intermediate rate results in Theorem~\ref{lemma:intermediate}.

\begin{definition}[$(v,c)$-Central Condition]\label{def:weaker-eta-c}
Let $v:[0, \infty) \rightarrow[0, \infty)$ be a bounded and non-decreasing function satisfying $v(\epsilon)>0$ for all $\epsilon > 0$. We say that $(\mu, \ell, \mathcal{W}, \mathcal{A})$  satisfies the $(v,c)$-central condition if for all $\epsilon \geq 0$, it holds that
\begin{align}
& \log \mathbb{E}_{P_W\otimes \mu}  \left[e^{-{v(\epsilon)}\left(\ell(W,Z)-\ell(w^*,Z)\right)}\right]  \leq \nonumber \\
& -cv(\epsilon)  \mathbb{E}_{P_W\otimes \mu}\left[\ell(W,Z) - \ell(w^*,Z)\right] + v(\epsilon) \epsilon. \label{eq:v-central} 
\end{align}
\end{definition}

\begin{theorem}\label{lemma:intermediate}
Assume the learning tuple $(\mu, \ell, \mathcal{W}, \mathcal{A})$ satisfies the  $\left(v, c\right)$-central condition up to $\epsilon$ for some function $v$ as defined in~Def. \ref{def:weaker-eta-c} and $0 < c < 1$. Then it holds that for any $\epsilon \geq 0$ and any $0< \eta \leq v(\epsilon)$,
\begin{align*}
     \mathbb{E}_{W\mathcal{S}_n}[\mathcal{E}(W,\mathcal{S}_n)] \leq & \frac{1-c}{c} \mathbb{E}_{P_{W\mathcal{S}_n}}[\hat{\mathcal{R}}\left(W, \mathcal{S}_{n} \right)] \\
     & + \frac{1}{n} \sum_{i=1}^{n} \left( \frac{1}{\eta c}I(W;Z_i) + \frac{\epsilon}{c}\right).
 \end{align*}
\end{theorem}

We prove this result in Appendix \ref{proof:intermediate}.  In particular, this result implies that if $v(\epsilon) \asymp \epsilon^{1-\beta}$ for some $\beta \in [0,1]$, then the generalization error is bounded by,
\begin{align*}
     \mathbb{E}_{W\mathcal{S}_n}&[\mathcal{E}(W,\mathcal{S}_n)] \leq  \frac{1-c}{c} \mathbb{E}_{P_{W\mathcal{S}_n}}[\hat{\mathcal{R}}\left(W, \mathcal{S}_{n} \right)] \\
     &+ \frac{(1-\beta)^{\frac{1}{2-\beta}}(2-\beta)}{nc(1-\beta)}\sum_{i=1}^{n} I(W;Z_i)^{\frac{1}{2-\beta}}.
 \end{align*}

 
Thus, the expected generalization is found to have an order of $I(W;Z_i)^{\frac{1}{2-\beta}}$, which corresponds to the results under Bernstein's condition \cite{hanneke2016refined,mhammedi2019pac,grunwald2021pac}.

\section{Examples}
In this section, we present several additional examples of fast rates that satisfy the $(\eta,c)$-central condition. In certain examples, the convergence rate may be faster than $O(1/n)$, for instance, when the mutual information $I(W;Z_i)$ exhibits an exponential convergence. In addition, we examine two supervised learning problems: linear regression, a simple problem, and logistic regression, a slightly more complicated problem, both of which demonstrate a convergence rate of $O(1/n)$. For completeness, we have also included some exceptional instances where the $(\eta,c)$-central condition is not satisfied, despite being relatively uncommon in learning problems.

\subsection{Gaussian Mean Estimation with Discrete $\mathcal{W}$}\label{sec:discretehypothesis}
In the Gaussian mean estimation problem, instead of considering $\mathcal{W} = \mathbb{R}$, we now consider an example where we assume $z_i \sim \mathcal{N}(w,1), i = 1,\cdots, n$ for some $w$ in the hypothesis space $\mathcal{W}$ that has finite elements, e.g., $\mathcal{W} = \{-1, 1\}$. We assume $w = 1$ is the true mean. Define the loss function to be the squared loss of $\ell(w,z) = (w-z)^2$ and define the empirical risk minimization problem
\begin{align}
    w_{\textup{ERM}} &= \argmin_{W \in \mathcal{W}} \frac{1}{n}\sum_{i=1}^{n}{(w - z_i)^2},
\end{align}
which is equivalent to the maximum likelihood estimation in this case. Similar to Example~\ref{eg:gaussian_mean_zero}, the ERM algorithm produces the decision rule by:
\begin{align*}
    w_{\ERM} = \begin{cases}
1, \textup{ if } \frac{1}{n}\sum_{i=1}^{n}Z_i \geq 0, \\
-1, \textup{ otherwise.}
\end{cases}
\end{align*} 
As $\frac{1}{n}\sum_{i=1}^{n}Z_i \sim \mathcal{N}(1, \frac{1}{n})$,  we then have that,
\begin{align*}
    P(W_{\ERM} = 1) &= 1 - Q(\sqrt{n}) \\
    & \leq 1 - \frac{\sqrt{n}}{1+n} \frac{1}{\sqrt{2\pi}} \exp(-\frac{n}{2}),
\end{align*}
where $Q(\cdot)$ is the tail distribution function of the standard normal distribution, and the inequality follows as $\frac{x}{1+x^2}\phi(x) \leq Q(x) < \frac{\phi(x)}{x}$ for any $x > 0$ where $\phi(x) = \frac{1}{\sqrt{2\pi}} e^{-\frac{x^2}{2}}$ \cite{borjesson1979simple}. Similarly, we have, 
\begin{align*}
    P(W_{\ERM} = -1) = Q(\sqrt{n}) &\leq \frac{1}{\sqrt{2\pi n}}  \exp(-\frac{n}{2}).
\end{align*}
Then we can calculate the expected risk and the excess risk by,
\begin{align*}
    \mathbb{E}_{P_W\otimes \mu}[\ell(W_{\ERM},Z)] &= P(w_{\ERM} = -1)\mathbb{E}_{Z}[(Z+1)^2] \\
    &\quad + P(w_{\ERM} = 1)\mathbb{E}_{Z}[(Z-1)^2]  \\
    &= 1 + 4Q(\sqrt{n}) \\
    & \leq 1 + \frac{4}{\sqrt{2\pi n}}\exp(-\frac{n}{2}).
\end{align*}
It is easily calculated that $w^* = 1$ in this case and 
\begin{align*}
    \mathbb{E}_{Z}[\ell(w^*,Z)] &= \mathbb{E}_{Z}[(Z-1)^2] = 1.
\end{align*}
Then the expected excess risk is calculated as,
\begin{align*}
    \mathbb{E}_{P_W\otimes \mu}& [\ell(W_{\ERM},Z)] - \mathbb{E}_{Z}[\ell(W^*,Z)]  = 4Q(\sqrt{n}) \\
    & \leq  \frac{4}{\sqrt{2\pi n}}\exp(-\frac{n}{2}) = O(\frac{e^{-\frac{n}{2}}}{\sqrt{n}}).
\end{align*}
The expected generalization error is given by,
\begin{align*}
    &\mathbb{E}_{P_W\otimes \mu}\left[\ell(W_{\ERM},Z_i)\right] - \frac{1}{n}\sum_{i=1}^{n}\mathbb{E}_{WZ_i}\left[\ell(W_{\ERM},Z_i)\right] \\
    & = 2\mathbb{E}_{W_{\ERM}\mathcal{S}_n}\left[W\frac{1}{n}\sum_{i=1}^{n}Z_i\right] - 2\mathbb{E}_W\left[W_{\ERM}\right]\\
    &= 2\mathbb{E}_{\mathcal{S}_n}\left[\left|\frac{1}{n}\sum_{i=1}^{n}Z_i\right|\right] - 2\mathbb{E}_W\left[W_{\ERM}\right] \\
    &= \sqrt{\frac{8}{n\pi}} \exp(-\frac{n}{2}) = O\left(\frac{e^{-\frac{n}{2}}}{\sqrt{n}}\right).
\end{align*}
where the first equality follows by expanding the squared loss $\ell(w,z)$, and the second equality follows as $W_{\ERM}$ has the same sign as $\sum^n_{i=1}Z_i$. The third equality is from the following calculation where:
\begin{align}
    \mathbb{E}_W[W_{\ERM}] = 1 - 2Q(\sqrt{n})
\end{align}
and 
\begin{align}
\mathbb{E}_{\mathcal{S}_n}\left[\left|\frac{1}{n}\sum_{i=1}^{n}Z_i\right|\right] = \sqrt{\frac{2}{n\pi}} \exp(-\frac{n}{2}) + 1 - 2Q(\sqrt{n})
\end{align}
as $\left|\frac{1}{n}\sum_{i=1}^{n}Z_i\right|$ is the folded Gaussian distribution with the parameters $\mu = 1$ and $\sigma = \frac{1}{\sqrt{n}}$. Now we evaluate our mutual information bound and compare it to the true excess risk and generalization error. For a fixed $z_i$, we have,
\begin{align*}
     P(W_{\ERM}=1|z_i) &= P(\frac{1}{n-1}\sum_{j\neq i} Z_j \geq - \frac{1}{n-1}z_i) \\
     &= 1 - Q(\frac{z_i}{\sqrt{n-1}} + \sqrt{n-1}).
\end{align*}
We then calculate the mutual information for a large $n$ by:
\begin{align*}
    &I(W_{\ERM};Z_i) = H(W_{\ERM}) - H(W_{\ERM}|Z_i) \\
    &= h_2(Q(\sqrt{n})) - \mathbb{E}_{Z_i}\left[h_2(Q(\frac{Z_i}{\sqrt{n-1}} + \sqrt{n-1}))\right] \\
    &= h_2(Q\sqrt{n}) \\
    &\quad - \mathbb{E}_{Z_i \leq -(n-1)}\left[h_2(Q(\frac{Z_i}{\sqrt{n-1}} + \sqrt{n-1}))\right] \\
    &\quad - \mathbb{E}_{Z_i > -(n-1)}\left[h_2(Q(\frac{Z_i}{\sqrt{n-1}} + \sqrt{n-1}))\right] \\
    &\leq h_2(Q(\sqrt{n})) - h_2(Q(\frac{\mathbb{E}_{Z_i > -(n-1)}[Z_i]}{\sqrt{n-1}} + \sqrt{n-1})) \\
    &= O\left(\frac{e^{-\frac{n}{2}}}{\sqrt{n}}\right)
\end{align*}
where the inequality follows Jensen's inequality with the fact that $Q$-function is locally convex for a large $n$ as $Q''(x) = \frac{x}{\sqrt{2\pi}}e^{-\frac{x^2}{2}}$ and $h_2$ is a concave function. 
By writing the excess risk as $r(w,z) = \ell(w,z) - \ell(w^*,z) = (w - z)^2 - (1 - z)^2$, its moment generating function is calculated by:
\begin{align*}
    &\mathbb{E}_{P_W\otimes \mu}[e^{-\eta r(W_{\ERM},Z)}] \\
    & = P(W_{\ERM}=1) + P(W_{\ERM}=-1) \mathbb{E}_{Z}[e^{-4\eta Z}] \\
    &= 1 - Q(\sqrt{n})  + Q(\sqrt{n})\mathbb{E}_{Z}[e^{-4\eta Z}].
\end{align*}
We can then calculate
\begin{align*}
    \mathbb{E}_{Z}[e^{-4\eta Z}] 
    &= e^{8\eta^2 - 4\eta}
\end{align*}
and 
\begin{align*}
    \log \mathbb{E}_{P_W\otimes \mu}[e^{-\eta r(W_{\ERM},Z)}] &= \log (1 -Q(\sqrt{n})  \\
    & + Q(\sqrt{n})\exp(8\eta^2 - 4\eta)).
\end{align*}
Therefore, we can check the $(\eta,c)$-central condition with $\log(1+x) \leq x$ for any $x> -1$:
\begin{align*}
    \log \mathbb{E}_{P_W\otimes \mu}[e^{-\eta r(W_{\ERM},Z)}] &\leq   Q(\sqrt{n}) (e^{8\eta^2-4\eta}-1) \\
    &= -4c\eta Q(\sqrt{n}).
\end{align*}
where we can choose $c = \frac{1-e^{8\eta^2-4\eta}}{4\eta} \leq 1-2\eta$ with selecting any $\eta \in (0,\frac{1}{2})$. Then the bound in Theorem \ref{thm:eta-c} shows that:
\begin{align*}
    \mathbb{E}_{W\mathcal{S}_n}[\mathcal{E}(W_{\ERM},\mathcal{S})] &\leq \frac{1}{c\eta n}\sum^{n}_{i=1}I(W_{\ERM};Z_i) \\
    & = O\left(\frac{e^{-\frac{n}{2}}}{\sqrt{n}}\right),
\end{align*}
which decays exponentially. Such a bound matches the true convergence rate, and this example shows that with the $(\eta,c)$-central condition, the convergence rate is dominated by the mutual information between the hypothesis and the instances.

\subsection{Linear Regression}
We now extend the Gaussian mean estimation example to the linear regression problem where we have the instance space $\mathcal{Z} = \mathcal{X} \times \mathcal{Y}$ where $\mathcal{X} \subseteq \mathbb{R}^{d}$ represents the feature space and $\mathcal{Y} \subseteq \mathbb{R}$ represents the label space. We consider the linear regression model with the case $\mathcal{X} \subseteq \mathbb{R}$ for simplicity such that the label is generated in the following way:
\begin{align}
    Y_i = w^*X_i + \epsilon_i
\end{align}
where $w^* \in \mathbb{R}$ denotes the underlying (but unknown) hypothesis and $\epsilon_i$ are the noises i.i.d. drawn from some zero-mean distribution. We consider the loss function $\ell(w,z_i) = (y_i - wx_i)^2$. Then the ERM solution can be calculated as:
\begin{align*}
    w_{\ERM} = \frac{\sum_{i=1}^{n} x_iy_i}{\sum_{i=1}^{n} x^2_i} = w^* + \sum_{i}\frac{x_j}{\sum_{j}x^2_j}\epsilon_i.
\end{align*}
We consider the fixed design such that $x_i$ is not randomized where the optimal hypothesis is $w^*$ and we assume $x_i \neq x_j$ for $i \neq j$. Assume $\epsilon_i \sim \mathcal{N}(0,\sigma^2)$, we can then calculate the expected loss over $\epsilon$ as:
\begin{align*}
    \mathbb{E}_{Z_i}[\ell(w^*,Z_i)] &= \mathbb{E}_{\epsilon_i}[(w^*x_i + \epsilon_i - w^*x_i)^2] \\
    &= \mathbb{E}_{\epsilon_i}[ \epsilon_i^2] \\
    &= \sigma^2.
\end{align*}
The expected loss under $w_{\ERM}$ can be calculated as:
\begin{align*}
    &\mathbb{E}_{P_W\otimes Z_i}[\ell(W_\ERM,Z_i)] 
    \\
    &= \mathbb{E}_{\epsilon_j \otimes \epsilon'_i}\left[(\epsilon'_i - \sum_{k}\frac{x_ix_k}{\sum_{j}x^2_j}\epsilon_k )^2\right]  \\
    &= \sigma^2 + \frac{x^2_i}{\sum_{j} x^2_j}\sigma^2    
\end{align*}
Therefore the excess risk can be calculated as:
\begin{align*}
    \frac{1}{n}\sum_{i=1}^{n}\mathbb{E}_{P_W\otimes \mu}[\ell(W_\ERM,Z_i)] -  \frac{1}{n}\sum_{i=1}^{n}&\mathbb{E}_{Z_i}[\ell(w^*,Z_i)] \\
    & = \frac{1}{n}\sigma^2,
\end{align*}
which scales with $O(\frac{1}{n})$. For the generalization error, we have,
\begin{align*}
    &\frac{1}{n}\sum_{i=1}^{n}\mathbb{E}_{P_W\otimes Z_i}[\ell(W_\ERM,Z_i)] \\
    &\quad - \frac{1}{n}\sum_{i=1}^{n}\mathbb{E}_{WZ_i}[\ell(W_\ERM,Z_i)] \\
    &= \sigma^2 + \frac{1}{n}\sigma^2 - \sigma^2 - \frac{1}{n}\sigma^2 + 2\frac{\sigma^2}{n} \\
    &= \frac{2}{n}\sigma^2,
\end{align*}
which also scales with $O(\frac{1}{n})$. Next, we can calculate the mutual information by:
\begin{align*}
    I(W;Z_i) &= h(W) - h(W|Z_i) = \frac{1}{2}\log \frac{\sum_{j}x_j^2}{\sum_{j \backslash i} x_j^2}.
\end{align*}
Then we have,
\begin{align*}
    \frac{1}{n}\sum_{i=1}^{n}I(W;Z_i) &\leq \frac{1}{2n} \sum_{i}\frac{x^2_i}{\sum_{j \backslash i}x^2_j}.
\end{align*}
For $n > 3$ and any $i$, there exists a constant $c \in (0,1)$ such that $\sum_{j \backslash i} x^2_j \geq c \sum_j x^2_j$ holds. Then we can further upper bound the mutual information terms by,
\begin{align*}
    \frac{1}{n}\sum_{i=1}^{n}I(W;Z_i) &\leq \frac{1}{2n} \sum_{i}\frac{x^2_i}{c\sum_{j}{x^2_j}} = \frac{1}{2nc},
\end{align*}
which scales with $O(1/n)$. Since $\ell(W,Z_i)$ is $(1+\frac{x^2_i}{\sum_j x^2_j})\sigma^2\chi^2$ distributed, we calculate the log-moment generating function as,
\begin{align}
    &\log \mathbb{E}_{P_W\otimes \mu}[e^{\eta \ell(W_\ERM,Z_i)}]  \nonumber \\
    & = -\frac{1}{2}\log(1 - 2(1+\frac{x^2_i}{\sum_j x^2_j})\sigma^2\eta).
\end{align}
Let $\sigma_{i} = (1+\frac{x^2_i}{\sum_j x^2_j})\sigma^2$. We have,
\begin{align}
    &\log\left(\mathbb{E}_{P_W\otimes \mu}[\exp{\eta (\ell(W_\ERM,Z_i)- \mathbb{E}[\ell(W_\ERM,Z_i)])} ]\right) \nonumber \\
    & \leq \sigma_i^4\eta^2, \textup{ for any } \eta < 0.
\end{align}
Therefore, following \cite{bu2020tightening}, the generalization error bound becomes:
\begin{align*}
    &\frac{1}{n}\sum_{i=1}^{n} \sqrt{2\sigma^4_iI(W;Z_i)} \\
    &= \frac{1}{n}\sum_{i=1}^{n} \sqrt{2\left(\frac{\sum_j x^2_j + x^2_i}{\sum_j x^2_j}\right)^2 \sigma^4I(W;Z_i)} \\
    &\leq \sqrt{(2-c)^2\sigma^4\frac{1}{nc}},
\end{align*}
which scales with $O(\sqrt{\frac{1}{n}})$. If we consider the excess risk $r(W,Z_i)$, we can calculate that,
\begin{align*}
    &\log \mathbb{E}_{P_W\otimes \mu}\left[\exp(-\eta r(W,Z_i) )\right] \\
    &= \frac{1}{2}\log \frac{\sum_{j} x^2_j / x^2_i}{\sum_{j} x^2_j / x^2_i - (4\eta^2\sigma^4 - 2\eta \sigma^2)} \\
    &\leq \frac{2\eta^2 \sigma^4 - \eta \sigma^2}{\sum_j x^2_j / x^2_i} \\
    &\leq -c\eta \frac{\sigma^2 x^2_i}{\sum_{j}x^2_j}.
\end{align*}
Hence the $(\eta,c)$ is satisfied for $c \leq 1 - 2\eta\sigma^2$ and $\eta \leq \frac{1}{4\sigma^2}$. By selecting $\eta = \frac{1}{4\sigma^2}$ and $c = \frac{1}{2}$, the bounds for the ERM (by ignoring the empirical excess risk terms) become:
\begin{align*}
    \mathbb{E}_{W\mathcal{S}_n}[\mathcal{E}(W,\mathcal{S}_n)] \leq \frac{1}{\eta c n} \sum_{i=1}^{n} I(W;Z_i) \\
    \leq \frac{8\sigma^2}{n} \sum_{i=1}^{n} I(W;Z_i) 
\end{align*}
which scales with $O(\frac{1}{n})$. 


\subsection{Logistic Regression}\label{sec:logistic}
We apply our bound in a typical classification problem. Consider a logistic regression problem in a 2-dimensional space. For each $w \in \mathbb{R}^2$ and $z_i = (x_i,y_i) \in \mathbb{R}^{2} \times \{0,1\}$, the loss function is given by
\begin{align*}
    \ell(w,z_i) := &-(y_i\log (\sigma(w^Tx_i)) \\
    &+ (1-y_i)\log (1 - \sigma(w^Tx_i)))
\end{align*}
where $\sigma(x) = \frac{1}{1+e^{-x}}$. Here each $x_i$ is drawn from a standard multivariate Gaussian distribution $\mathcal{N}(0,\mathbf{I}_{2})$ and Let $w^* = (0.5,0.5)$, then each $y_i$ is drawn from the Bernoulli distribution with the probability $P(Y_i = 1|x_i, w^*) = \sigma(-x^T_iw^*)$. We also restrict hypothesis space as $\mathcal{W} = \{w: \|w\|_2 < 3\}$ where $W_{\ERM}$ falls in this area with high probability. Since the hypothesis is bounded and under the log-loss, then the learning problem will satisfy the central and witness condition \cite{van2015fast,grunwald2020fast}. Therefore, it will satisfy the $(\eta,c)$-central condition. We will evaluate the generalization error and excess risk bounds in (\ref{thm:eta-c}). To this end, we need to estimate $\eta$, $c$ and mutual information $I(W_{\ERM},Z_i)$ efficiently. Hence we repeatedly generate $w_{\ERM}$ and $z_i$ and use the empirical density for estimation. Specifically, we vary the sample size $n$ from $50$ to $500$, and for each $n$, we repeat the logistic regression algorithm 500 times to generate a set of $w_\ERM$. By fixing $\eta = 0.8$ as an example, we can empirically estimate the CGF and expected excess risk with the data sample and a set of ERM hypotheses, which leads to $c \approx 0.385$. It is worth noting that from the experiments, once $\eta$ is fixed, the choice of $c$ actually does not depend on the sample size $n$, which empirically confirms the ($\eta,c$)-central condition. For the mutual information, we decompose $I(W; X,Y) =I(W; Y) + P(Y = 0)I(W; X|y = 0) + P(Y = 1)I(W; X|y = 1)$ by chain rule, and the first term can be approximated using the continuous-discrete estimator~\cite{gao2017estimating} for mutual information and the rest terms are continuous-continuous ones~\cite{moddemeijer1989estimation,kraskov2004estimating}. To demonstrate the usefulness of the results, we also compare the bounds with the true excess risk and true generalization error. The comparisons are shown in Figure~\ref{fig:logistic}. We point out that the mutual information estimate, in this case, is not a trivial task as it involves the mutual information between continuous and discrete random variables. As shown in Figure \ref{fig:logistic}, both the true generalization error and the excess risk converge as $O(\frac{1}{n})$. The bounds on the generalization error and the excess error also seem to follow the same convergence rate. However, a more rigorous study of the estimated mutual information term is required to give a conclusive statement. Furthermore, we also plot the comparison with the bound in \cite{bu2020tightening} with the form of $\frac{1}{n}\sum_{i=1}^{n}\sqrt{\frac{I(W;Z_i)}{2}}$, and the comparison shows that our bound is decaying faster as $n$ increases.

\begin{figure*}[!ht]
\centering
\subfloat[Generalization Error]{\includegraphics[width=0.3\textwidth]{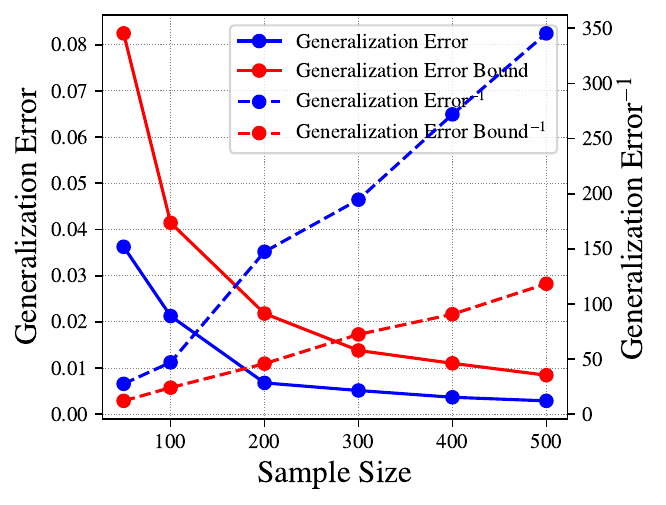}}\quad
\subfloat[Excess Risk]{\includegraphics[width=0.3\textwidth]  
{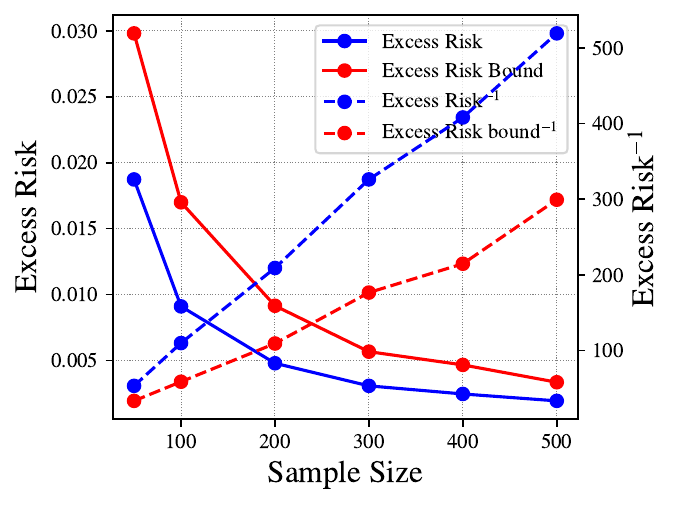}}\quad
\subfloat[Bound Comparison]{\includegraphics[width=0.26\textwidth]  
{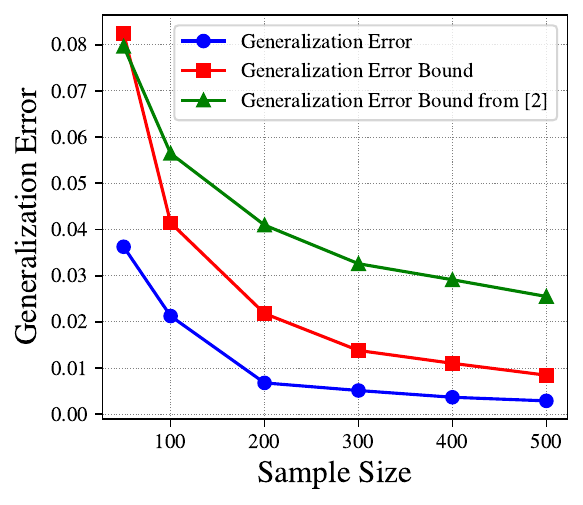}}
\caption{We represent the true expected generalization error in (a) and true excess risk in (b) along with their bounds in Theorem~\ref{thm:eta-c}. Here we vary $n$ from 50 to 500. We also plot their reciprocals to show the rate w.r.t. sample size $n$. We also plot the generalization error bound comparisons with \cite{bu2020tightening} in (c). All results are derived by 500 experimental repeats.}\label{fig:logistic}
\end{figure*}



\subsection{Examples when $(\eta,c)$-central condition does not hold}
There are cases where $(\eta,c)$-central condition does not hold. One instance is when $\mathbb{E}_{P_W\otimes \mu}[r(W,Z)] = 0$ but $\log\mathbb{E}_{P_W\otimes \mu}[e^{-\eta r(W,Z)}] > 0$. Notice that as $\mathbb{E}_{\mu}[{r(w,Z)}]$ is non-negative, hence requiring $\mathbb{E}_{P_W\otimes \mu}[r(W,Z)] = 0$ is the same as requiring $\mathbb{E}_{\mu}[{r(w,Z)}]=0$ for all $w$, namely all the hypotheses given by the algorithm are as good as the optimal hypothesis $w^*$. On the other hand, the CGF is not necessarily equal to zero in this case. We point out that though the exact condition $\mathbb{E}_{\mu}[{r(w,Z)}]=0$ for all $w$ seems rarely hold, it may be relevant in deep learning scenarios where one can often find many empirical minimizers that have a very low excess risk. In the following, we give a simple example where this condition holds (hence $(\eta,c)$-central condition does not hold) for the completeness of the results.


\begin{example}[Gaussian mean estimation with all hypotheses equally optimal] \label{eg:gaussian_mean_zero}
Different from the example in \ref{sec:discretehypothesis}, we consider the zero mean Gaussian case where $Z_i \sim \mathcal{N}(0,1), i = 1,\cdots, n$ and the same hypothesis space $\mathcal{W} = \{-1, 1\}$. Intuitively speaking, unlike the previous situation, even if we have abundant data, both $1$ and $-1$ are the optimal hypothesis due to the symmetry. Mathematically, let the algorithm be maximum likelihood estimation and the ERM algorithm produces the decision rule by:
\begin{align*}
W_{\ERM} = \begin{cases}
1, \textup{ if } \frac{1}{n}\sum_{i=1}^{n}Z_i \geq 0, \\
-1, \textup{ otherwise.}
\end{cases}
\end{align*} 
Then the distribution of $W_{\ERM}$ is Bernoulli distributed with $P(W_{\ERM} = 1) = P\left(\frac{1}{n}\sum_{i=1}^{n}Z_i \geq 0 \right).$ As $\frac{1}{n}\sum_{i=1}^{n}Z_i \sim \mathcal{N}(0, \frac{1}{n})$,  we have that,
\begin{align*}
    P(W_{\ERM} = 1) &= P(W_{\ERM} = -1) = \frac{1}{2} 
\end{align*}
due to the symmetry. Then we can calculate the expected risk and the excess risk by,
\begin{align*}
&\mathbb{E}_{P_W\otimes \mu}[\ell(W_{\ERM},Z)] \\
& = P(W_{\ERM} = 1)\mathbb{E}_{Z}[(Z-1)^2] \\
&\quad + P(W_{\ERM} = -1)\mathbb{E}_{Z}[(Z+1)^2] \\
&= 2.
\end{align*}
It is important to note that $w^* $ can be either $-1$ or $1$ in this case, and 
\begin{align*}
    \mathbb{E}_{Z}[\ell(W^*,Z)] &= \mathbb{E}_{Z}[(Z-1)^2] = 2.
\end{align*}
Then the \textbf{excess risk} is calculated as,
\begin{align*}
    \mathbb{E}_{P_W\otimes \mu}[\ell(W_{\ERM},Z)] - \mathbb{E}_{Z}[\ell(W^*,Z)]  = 0.
\end{align*}
The expected generalization error is given by,
\begin{align*}
    &\mathbb{E}_{P_W\otimes \mu}[\ell(W_{\ERM},Z_i)] - \frac{1}{n}\sum_{i=1}^{n}\mathbb{E}_{WZ_i}[\ell(W_{\ERM},Z_i)] \\  
    &= 2\mathbb{E}_{\mathcal{S}_n}\left[\left|\frac{1}{n}\sum_{i=1}^{n}Z_i\right|\right] 
    = \sqrt{\frac{8}{\pi n}}
\end{align*}
due to the fact that the expectation of the absolute Gaussian r.v. with mean of $\mu$ and variance of $\sigma^2$ is $\sigma \sqrt{\frac{2}{\pi}} e^{-\frac{\mu^2}{2 \sigma^2}}+\mu\left(1-2 \Phi\left(\frac{-\mu}{\sigma}\right)\right)$.  
If we choose $w^* = 1$ and the excess risk can be written as $r(w,z) = \ell(w,z) - \ell(w^*,z) = (w - z)^2 - (1 - z)^2$, now we can calculate the MGF of the excess risk explicitly as: 
\begin{align*}
    &\mathbb{E}_{P_W\otimes \mu}[e^{-\eta r(W_{\ERM},Z)}] \\
    &\quad = P(W_{\ERM}=1) + P(W_{\ERM}=-1) \mathbb{E}_{Z}[e^{-4\eta Z}] \\
    &= \frac{1}{2}+ \frac{1}{2}\mathbb{E}_{Z}[e^{-4\eta Z}].
\end{align*}
For the second term, we have the following:
\begin{align*}
    \mathbb{E}_{Z}[e^{-4\eta Z}] &= \int^{+\infty}_{-\infty} \frac{1}{\sqrt{2\pi}}e^{-\frac{z^2}{2}}e^{-4\eta z} dz 
    = e^{8\eta^2}.
\end{align*}
Finally, it yields the CGF by, 
\begin{align*}
    \log \mathbb{E}_{P_W\otimes \mu}[e^{-\eta r(W,Z)}] = \log ( \frac{1}{2}  + \frac{1}{2}\exp(8\eta^2))
\end{align*}
for any $\eta$. In this case, the $(\eta,c)$-central condition does \textbf{not} hold as $\log \mathbb{E}_{P_W\otimes \mu}[e^{-\eta r(W,Z)}] = \log ( \frac{1}{2}  + \frac{1}{2}\exp(8\eta^2)) \geq 0$ for any $\eta > 0$. 
\end{example}

\begin{example}[Hypothesis Selection \cite{russo2016controlling}] \label{eg:hypothesis_selection}
We provide another example where the $(\eta,c)$-central condition is not satisified. Let $\mathcal{S} = (Z_1,Z_2,\cdots, Z_n)$ where $Z_i \in \mathbb{R}$ is a random variable with the mean of $\mu_i$ and $W$ to be some selection hypothesis such that $\mathcal{W} = [1:n]$ where $[1:n]$ denotes a range of integers from 1 to n. If we seek for the largest instance $Z_i$ in the dataset $\mathcal{S}$, we define the loss function $\ell(W,\mathcal{S}) = -Z_{W}$ and simply minimising the loss function will lead to a hypothesis that produces the instance with the largest value. Then we have that $\mathbb{E}_{W\mathcal{S}}[\ell(W,\mathcal{S})] = -\mathbb{E}_{W\mathcal{S}}[Z_W]$ and $\mathbb{E}_{W\otimes \mathcal{S}}[\ell(W,\mathcal{S})] = -\mathbb{E}_{W\otimes \mathcal{S}}[Z_W] = -\mathbb{E}_{W}\mathbb{E}_\mathcal{S}[Z_W] = -\mathbb{E}_W[\mu_W]$. We can write:
\begin{align*}
    |\mathbb{E}_{W\mathcal{S}}[\ell(W,\mathcal{S})] - \mathbb{E}_{W\otimes \mathcal{S}}[\ell(W,\mathcal{S})]| = |\mathbb{E}_{W\mathcal{S}}[Z_W - \mu_W]|. 
\end{align*}
We consider the ERM hypothesis $W = \argmin \ell(W,\mathcal{S}) = \argmax_i Z_i$ and assume all $Z_i$ are normally distributed and have the same mean of some positive $\mu$ and variance of $\sigma^2$ for simplicity. Then we have that:
\begin{align*}
I(W;\mathcal{S}) = H(W) - H(W|\mathcal{S}) = H(W) = \log n.    
\end{align*}
Since we assume $\mu_W = \mu$ for all $W$, we have $|\mathbb{E}_{W\mathcal{S}}[Z_W - \mu_W]| \leq \sqrt{2\sigma^2 \log n} = \sqrt{2\sigma^2 I(W;\mathcal{S})}$ for any $n$. Now we check the $(\eta,c)$-central condition w.r.t. the excess risk function $r(W,\mathcal{S}) = \ell(W,\mathcal{S}) - \ell(w^*, \mathcal{S})$ where $w^*$ can be any index in $[1:n]$:
\begin{align*}
    \log \mathbb{E}_{W\otimes \mathcal{S}}\left[ e^{-\eta r(W, \mathcal{S})} \right] \leq -c\eta \mathbb{E}_{W\otimes \mathcal{S}}\left[ r(W, \mathcal{S})\right]
\end{align*}
for some $c \in [0,1]$. We can then calculate that 
\begin{align*}
    \log \mathbb{E}_{W\otimes \mathcal{S}}\left[ e^{-\eta r(W, \mathcal{S})}\right] &= \log \left(\frac{1}{n} + \frac{n-1}{n} \exp{\sigma^2\eta^2}\right) \\
    &> 0 ,
\end{align*}
for positive $\eta$ while the expected excess risk is $\mathbb{E}_{W\otimes \mathcal{S}}\left[ r(W, \mathcal{S})\right] = 0$ and this contradicts to the proposed $(\eta,c)$-central condition. In both examples, we observe that the expected excess risk for the ERM hypothesis is always zero because each hypothesis in the class is optimal for the given distribution. However, the CGF could still be positive since we have stochasticity in the hypothesis derived, which is a rare occurrence in most learning scenarios, and we highlight this unique scenario for completeness.  
\end{example}

\section{Conclusion}
As we demonstrate in this paper, if the sub-Gaussian assumption is imposed on the excess loss in the typical information-theoretic generalization error bounds, the square root does not necessarily prevent us from achieving a fast rate. On top of that, we identify the key conditions that lead to the fast and intermediate rate bounds in expectation. Intuitively speaking, to achieve a fast rate bound for both the generalization error in expectation, the output hypothesis of a learning algorithm must be ``good" enough compared to the optimal hypothesis $w^*$. Here we encode the notion of goodness in terms of the CGF by controlling the gap between $\ell(w,Z)$ and $\ell(w^*,Z)$.  Further, we verify the proposed bounds and present the results analytically with examples. We remark that there is some room for future work to improve the bounds. One possible direction is to develop novel techniques for removing the empirical excess risk term in the bound for general algorithms other than ERM or regularized ERM. Additionally, in most cases, $w^*$ is usually not known but one could possibly seek a reference hypothesis $\hat{w}$ that could be trained from the sample and close to $w^*$, allowing more flexibility of the bounds in real applications.

\section*{Acknowledgement}
The preliminary version of this work is presented at the ITW2022 conference, we greatly appreciate useful feedback and comments from all reviewers. The work  is supported by Melbourne Research Scholarships (MRS).

\appendices

\section{Proof of~Theorem~\ref{thm:sub-Gaussian}}\label{proof:sub-Gaussian}
\begin{proof}
As $w^*$ is independent of $Z_i$, we have
\begin{align}
    &\mathbb{E}_{W\mathcal{S}_n}[\mathcal{E}(W, \mathcal{S}_n)] = \mathbb{E}_{W\mathcal{S}_n}[\mathcal{R}(W) - \hat{\mathcal{R}}(W, \mathcal{S}_n)] \nonumber \\
    &= \mathbb{E}_{W \otimes \mathcal{S}_n}[\hat{\mathcal{R}}(W,\mathcal{S}_n)] - \mathbb{E}_{W\mathcal{S}_n}[\hat{\mathcal{R}}(W,\mathcal{S}_n)] .
\end{align}
Let the distribution $P_{WZ_i}$ denote the joint distribution induced by $P_{W\mathcal{S}_n}$ with the algorithm $P_{W|\mathcal{S}_n}$. With the i.i.d. assumption, we can rewrite the generalization error by:
\begin{align*}
    \mathbb{E}_{W\mathcal{S}_n}[\mathcal{E}(W,\mathcal{S}_n)] = &\frac{1}{n}\sum_{i=1}^n \mathbb{E}_{P_W \otimes \mu}[r(W,Z_i)] \\
    &- \mathbb{E}_{WZ_i}[r(W,Z_i)].
\end{align*}
Using the KL-divergence property \cite{bu2020tightening, xu2017information} that
\begin{align}
    & \sqrt{ 2\sigma^2 D\left( P_{WZ_i} \| P_{W}\otimes P_{Z_i} \right)} \nonumber \\
    & \geq \mathbb{E}_{P_W \otimes \mu}[r(W,Z_i)] - \mathbb{E}_{WZ_i}[r(W,Z_i)]
\end{align}
under the $\sigma$-sub-Gaussian assumption under the distribution $P_{W} \otimes \mu$. Summing up every term concludes the proof.
\end{proof}

\section{Proof of~Theorem~\ref{thm:lower_bounds}}\label{proof:lowerbound}
\begin{proof}
By Jensen's inequality, we have the inequalities as in (\ref{eq:lower_bounds}).
\begin{figure*}[!htb]
\normalsize
\begin{align}
    I(W;Z_i) &= \mathbb{E}_{P_{WZ_i}}\left[-\frac{r(W,Z_i)}{2\sigma_N^2}- \frac{(W-Z_i)^2}{2\sigma^2_N(n-1)}\right] - \log \mathbb{E}_{P_W\otimes \mu}\left[\exp{\left(-\frac{r(W,Z_i)}{2\sigma_N^2} - \frac{(W-Z_i)^2}{2\sigma^2_N(n-1)}\right)}\right] \nonumber \\
    & \leq \mathbb{E}_{P_{WZ_i}}\left[-\frac{r(W,Z_i)}{2\sigma_N^2}- \frac{(W-Z_i)^2}{2\sigma^2_N(n-1)}\right] - \mathbb{E}_{P_W\otimes \mu}\left[{-\frac{r(W,Z_i)}{2\sigma_N^2} - \frac{(W-Z_i)^2}{2\sigma^2_N(n-1)}}\right], \label{eq:lower_bounds}
\end{align}
\vspace*{4pt}
\hrulefill
\end{figure*}
By rearranging the inequality,  the following holds for all $i = 1,2, \cdots, n$, 
\begin{align}
    &\mathbb{E}_{P_W\otimes \mu}\left[r(W,Z_i)\right] \geq 2\sigma^2_NI(W;Z_i) \nonumber \\
    &+ \mathbb{E}_{P_{WZ_i}}[r(W,Z_i)] \nonumber \\
    &+ \frac{1}{n-1} (\mathbb{E}_{WZ_i}\left[ \ell(W,Z_i)\right] - \mathbb{E}_{P_W\otimes \mu}\left[ \ell(W,Z_i)\right]).
\end{align}
We complete the proof of the lower bound on the excess risk by averaging over $Z_i$. For the generalization error, we rewrite \eqref{eq:lower_bounds} as
\begin{align}
   \frac{1}{2\sigma^2_N}\frac{n}{n-1} (\mathbb{E}_{P_W\otimes \mu}\left[ \ell(W,Z_i)\right] &- \mathbb{E}_{WZ_i}\left[ \ell(W,Z_i)\right] ) \nonumber \\
   & \geq I(W;Z_i).
\end{align}
Summing up every term for $Z_i$ completes the proof for the generalization error.
\end{proof}

\section{Proof of~Theorem~\ref{thm:sub-Gaussianv2}}\label{proof:sub-Gaussianv2}
\begin{proof}
Using the Donsker-Varadhan representation of the KL divergence, we build on the following inequality for some $\eta > 0$,
\begin{align}
    \frac{I(W;Z_i)}{\eta} &+ \Esub{P_{WZ_i}}{r(W,Z_i)} \nonumber \\
    &\geq  - \frac{1}{\eta} \log \mathbb{E}_{P_{W}\otimes \mu}[e^{-\eta(r(W,Z_i))}].
\end{align}
We will bound the R.H.S. using the following technique. For any $a$, we have,
\begin{align}
    &\log \mathbb{E}_{P_{W}\otimes \mu}\left[e^{a\eta \bar{r} -\eta r(W,Z_i)}\right] \nonumber \\
    &\quad = \log \mathbb{E}_{P_{W}\otimes \mu}\left[e^{\eta (\bar{r} - r(W,Z_i)) + (a-1)\eta \bar{r}}\right] \\
    &\quad \leq \frac{\sigma^2 \eta^2}{2} + (a-1) \eta \bar{r}.
\end{align}
which is equivalent to
\begin{align}
    \log \mathbb{E}_{P_{W}\otimes \mu}[e^{-\eta(r(W,Z_i))}] \leq \frac{\sigma^2 \eta^2}{2} - \eta \mathbb{E}[r(W,Z_i)],
\end{align}
By defining $0 < a_\eta = 1-  \frac{\eta\sigma^2}{2\mathbb{E}[r(W,Z_i)]} < 1$, the R.H.S. of the expression can be written as:
\begin{align}
    \frac{\sigma^2 \eta^2}{2} - \eta \mathbb{E}[r(W,Z_i)] = -a_\eta \eta \mathbb{E}[r(W,Z_i)].
\end{align}
Hence we have
\begin{align}
    \log \mathbb{E}_{P_{W}\otimes \mu}[e^{-\eta(r(W,Z_i))}] &\leq -a_\eta \eta \mathbb{E}[r(W,Z_i)]. 
\end{align}
where $0 < \eta < \frac{2\mathbb{E}[r(W;Z_i)]}{\sigma^2}$. We then rewrite the inequality as,
\begin{align}
    -\frac{1}{\eta}\log \mathbb{E}_{P_{W}\otimes \mu}[e^{-\eta(r(W,Z_i))}] &\geq a_\eta \mathbb{E}[r(W,Z_i)], 
\end{align}
and we further have the following bound,
\begin{align}
    \frac{I(W;Z_i)}{\eta} + \mathbb{E}_{P_{WZ_i}}[r(W,Z_i)] &\geq a_\eta \mathbb{E}_{P_WP_{Z_i}}[r(W,Z_i)].
\end{align}
Hence,
\begin{align}
    &\mathbb{E}_{P_WP_{Z_i}}[r(W,Z_i)] - \mathbb{E}_{P_{WZ_i}}[r(W,Z_i)] \nonumber \\
    &\leq \frac{I(W;Z_i)}{\eta a_\eta} + \frac{1-a_\eta}{a_\eta}\mathbb{E}_{P_{WZ_i}}[r(W,Z_i)].
\end{align}
Summing every term for $Z_i$, we have,
\begin{align}
     \mathbb{E}_{W\mathcal{S}_n} \left[\mathcal{E}(W, \mathcal{S}_n)\right] \leq & \frac{1-a_\eta}{a_\eta} \mathbb{E}_{W\mathcal{S}_n}[\hat{\mathcal{R}(W,\mathcal{S}_n)}] \nonumber \\
     &+ \frac{1}{n\eta a_\eta} \sum_{i=1}^{n}  I\left(W ; Z_{i}\right),
\end{align}
which completes the proof.
\end{proof}

\section{Proof of Theorem~\ref{thm:eta-c}}\label{proof:thm_eta-c}
\begin{proof}
Firstly we rewrite excess risk and empirical excess risk by:
    \begin{align}
        \mathcal{R}(w) &= \mathbb{E}_{Z\sim \mu}[\ell(w,Z)] - \mathbb{E}_{Z\sim \mu}[\ell(w^*,Z)] \nonumber \\
        &= \frac{1}{n}\sum_{i=1}^{n} \mathbb{E}_{Z_i\sim \mu}[\ell(w,Z_i)] - \mathbb{E}_{Z_i \sim \mu}[\ell(w^*,Z_i)] \nonumber \\
        &= \mathbb{E}_{\mathcal{S}_n}[\hat{\mathcal{R}}(w, \mathcal{S}_n)],
    \end{align}
    and 
    \begin{align}
        \hat{\mathcal{R}}(w, \mathcal{S}_n) &=  \hat{L}(w,\mathcal{S}_n) - \hat{L}(w^*,\mathcal{S}_n).
    \end{align}
Given any $\mathcal{S}_n$, the gap between the excess risk and empirical excess risk can be written as,
    \begin{align}
    \mathcal{R}(w) - \hat{\mathcal{R}}(w, \mathcal{S}_n) = \mathbb{E}_{\mathcal{S}_n}[\hat{\mathcal{R}}(w, \mathcal{S}_n)] - \hat{\mathcal{R}}(w, \mathcal{S}_n).
    \end{align}
We will bound the above quantity by taking the expectation w.r.t. $W$ learned from $\mathcal{S}_n$ by: 
    \begin{align}
        &\mathbb{E}_{W\mathcal{S}_n}[\mathcal{E}(W)] = \mathbb{E}_{W\mathcal{S}_n}[\mathcal{R}(W) - \hat{\mathcal{R}}(W, \mathcal{S}_n)] \\
        &=  \mathbb{E}_{P_W\otimes \mathcal{S}_n}[\hat{\mathcal{R}}(W, \mathcal{S}_n)] - \mathbb{E}_{W\mathcal{S}_n}[\hat{\mathcal{R}}(W, \mathcal{S}_n)] \\
        &= \frac{1}{n}\sum_{i=1}^n \mathbb{E}_{P_W \otimes \mu}[r(W,Z_i)] - \mathbb{E}_{WZ_i}[r(W,Z_i)].
    \end{align}
Recall that the variational representation of the KL divergence between two distributions $P$ and $Q$ defined over $\mathcal X$ is given as (see, e. g. \cite{boucheron2013concentration})
\begin{align}
D(P||Q)=\sup_{f}\{\Esub{P}{f(X)}-\log\Esub{Q}{e^{f(x)}} \}, 
\end{align}
where the supremum is taken over all measurable functions such that $\Esub{Q}{e^{f(x)}}$ exists. Let $f(w,z_i) = -\eta r(w,z_i)$, we have inequality as shown in~(\ref{eq:MI_KL}).
\begin{figure*}[!htb]
\normalsize
\begin{align}
    D(P_{WZ_i}\|P_{W}\otimes P_{Z_i}) &\geq \Esub{P_{WZ_i}}{-\eta r(W,Z_i)} - \log \mathbb{E}_{P_{W}\otimes \mu}[e^{-\eta(r(W,Z_i))}] \nonumber \\
    &= \Esub{P_{WZ_i}}{-\eta r(W,Z_i)} - \log \mathbb{E}_{P_{W}\otimes \mu}[e^{-\eta(r(W,Z_i) - \mathbb{E}_{P_{W}\otimes \mu}[r(W,Z_i)])  }] + \Esub{P_W \otimes \mu}{\eta r(W,Z_i)} \nonumber \\
    &= \eta\left(\Esub{P_W \otimes \mu}{r(W,Z_i)} - \Esub{P_{WZ_i}}{r(W,Z_i)}\right) - \log \mathbb{E}_{P_{W}\otimes \mu}[e^{\eta( \mathbb{E}_{P_{W}\otimes \mu}[r(W,Z_i)]  - r(W,Z_i))  }]. \label{eq:MI_KL}
\end{align}
\vspace*{4pt}
\hrulefill
\end{figure*}
Next we will upper bound the second term $\log \mathbb{E}_{P_{W}\otimes \mu}[e^{\eta( \mathbb{E}_{P_{W}\otimes \mu}[r(W,Z_i)]  - r(W,Z_i))  }]$ in R.H.S. using the expected $(\eta,c)$-central condition. From the $(\eta, c)$-central condition, we have,
\begin{align}
    &\log \mathbb{E}_{P_{W}\otimes \mu}[e^{\eta( \mathbb{E}_{P_{W}\otimes \mu}[r(W,Z_i)]  - r(W,Z_i))  }] \nonumber \\
    & \leq (1-c)\eta \mathbb{E}_{P_{W}\otimes \mu}[r(W,Z_i)].
\end{align}
Then we arrive at,
\begin{align}
    I(W;Z_i) &\geq \eta \left(\Esub{P_W \otimes \mu}{r(W,Z_i)} - \Esub{P_{WZ_i}}{r(W,Z_i)}\right) \nonumber  \\
    &\quad - (1-c)\eta \mathbb{E}_{P_{W}\otimes \mu}[r(W,Z_i)].
\end{align}
Divide $\eta$ on both side; we arrive at,
\begin{align}
    \frac{I(W;Z_i)}{\eta} 
    &\geq \Esub{P_W \otimes \mu}{r(W,Z_i)} - \Esub{P_{WZ_i}}{r(W,Z_i)} \nonumber \\
    &\quad - (1-c) \mathbb{E}_{P_{W}\otimes \mu}[r(W,Z_i)].
\end{align}
Rearrange the equation and yields,
\begin{align}
    c\Esub{P_W \otimes \mu}{r(W,Z_i)} \leq \Esub{P_{WZ_i}}{r(W,Z_i)} + \frac{I(W;Z_i)}{\eta}.
\end{align}
Therefore,
\begin{align}
    &\Esub{P_W \otimes \mu}{r(W,Z_i)} - \Esub{P_{WZ_i}}{r(W,Z_i)} \nonumber \\
    &\quad \leq \left(\frac{1}{c} - 1\right)\mathbb{E}_{P_{WZ_i}}[r(w,Z_i)] + \frac{I(W;Z_i)}{c\eta}.
\end{align}
Summing up every term for $Z_i$ and divide by $n$, we end up with,
\begin{align}
    &\Esub{P_W \otimes P_{\mathcal{S}_n}}{\hat{\mathcal{R}}(W,\mathcal{S}_n)} - \Esub{P_{W\mathcal{S}_n}}{\hat{\mathcal{R}}(W,\mathcal{S}_n)} \nonumber \\
    &\leq  (\frac{1}{c} - 1) \left({\mathbb{E}_{P_{W\mathcal{S}_n}}}[\hat{\mathcal{R}}(W,\mathcal{S}_n)]\right) + \frac{1}{n}\sum_{i=1}^{n}\frac{I(W;Z_i)}{c\eta}.
\end{align}
Finally, we complete the proof by,
\begin{align}
    \mathbb{E}_{W\mathcal{S}_n}[\mathcal{E}(W)] \leq &(\frac{1}{c} - 1) {\mathbb{E}_{P_{W\mathcal{S}_n}}}[\hat{\mathcal{R}}(W,\mathcal{S}_n)] \nonumber \\
    & + \frac{1}{n}\sum_{i=1}^{n}\frac{I(W;Z_i)}{c\eta}.
\end{align}
\end{proof}

\section{Proof of~Theorem~\ref{thm:eta-c-loss}}
\begin{proof}
Due to the Donsker-Varadhan representation, we have that for each $Z_i$:
\begin{align}
    I(W;Z_i) 
    &\geq -\mathbb{E}_{WZ_i}[\eta \ell(W, Z_i)] \nonumber \\
    &\qquad - \log \mathbb{E}_{P_W\otimes \mu}\left[e^{-\eta \ell(W,Z_i)} \right] \nonumber \\
    &\geq -\mathbb{E}_{WZ_i}[\eta \ell(W, Z_i)] \nonumber \\
    &\qquad + c\eta \mathbb{E}_{P_W \otimes \mu}[\ell(W,Z)].
\end{align}
The last inequality holds due  the $(\eta,c)$-central condition. By rearranging the equation, we arrive at the bound for the generalization error as:
\begin{align}
    &\mathbb{E}_{P_W \otimes \mu}[\ell(W,Z_i)] -\mathbb{E}_{WZ_i}[\ell(W, Z_i)] \nonumber \\
    & \leq \frac{I(W;Z_i)}{c\eta} + \frac{1-c}{c}\mathbb{E}_{WZ_i}[\ell(W,Z_i)]
\end{align}
which completes the proof by:
\begin{align}
    \mathbb{E}_{W\mathcal{S}_n}[\mathcal{E}(W,\mathcal{S}_n)] \leq & \frac{\sum_{i=1}^{n}I(W;Z_i)}{c\eta n} \nonumber \\
    & + \frac{1-c}{c}\mathbb{E}_{W\mathcal{S}_n}[\hat{L}(W,\mathcal{S}_n)].
\end{align}
\end{proof}

\section{Proof of~Corollary~\ref{coro:berstein}}\label{proof:corollary_bernstein}
\begin{proof}
We first present the expected Bernstein inequality which will be the key technical lemma for the fast rate bound.
\begin{lemma}[Expected Bernstein Inequality \cite{mhammedi2019pac,cesa2006prediction}] \label{lemma:exp_bern}
Let $U$ be a random variable bounded from below by $-b < 0 $ almost surely, and let $\kappa(x)=(e^x - x - 1) / x^{2} .$ For all $\eta >0$, we have
\begin{align}
 \log \mathbb{E}_{U}\left[e^{\eta(\mathbb{E}[U]-U}) \right] \leq \eta^2 c_{\eta} \cdot \mathbb{E}[U^{2}],
\end{align}
for all $c_{\eta} \geq  \kappa(\eta b)$.
\end{lemma}
\begin{proof}
The proof of the lemma follows from~\cite{cesa2006prediction} and \cite{mhammedi2019pac}. Firstly we define $Y = - U$ which is upper bounded by $b$, then using the property that $\frac{e^Y - Y -1}{Y^2}$ in non-increasing for $Y \in \mathbb{R}$, then we define $Z = \eta Y$ such that,
\begin{align}
   \frac{e^{Z} - Z -1}{Z^2} \leq \frac{e^{\eta b} - \eta b - 1}{\eta^2 b^2} = \kappa (\eta b).
\end{align}
Rearranging the inequality, we then arrive at,
\begin{align}
   e^{Z} - Z -1 \leq Z^2 \kappa (\eta b).
\end{align}
Taking the expectation on both sides, we have:
\begin{align}
    \mathbb{E}[e^{Z}] - 1 \leq \mathbb{E}[Z^2]\kappa (\eta b) + \mathbb{E}[Z]
\end{align}
and using the fact that $\log (x) \leq x - 1$ for any $x \in (0, +\infty)$, we have,
\begin{align}
    \log\mathbb{E}[e^{Z}] \leq \mathbb{E}[e^{Z}] - 1 \leq \mathbb{E}[Z^2]\kappa (\eta b) + \mathbb{E}[Z].
\end{align}
Then we arrive at:
\begin{align}
   \log\mathbb{E}[e^{Z}]  - \mathbb{E}[Z] \leq \mathbb{E}[Z^2] \kappa (\eta b).
\end{align}
By substituting $Z = \eta Y$, we have
\begin{align}
   \mathbb{E}[e^{\eta (Y - \mathbb{E}[Y])}] \leq e^{\eta^2 \mathbb{E}[Y^2] \kappa (\eta b)}.
\end{align}
Define $c_\eta \geq \kappa(\eta b)$, it yields that
\begin{align}
   \mathbb{E}[e^{\eta (Y - \mathbb{E}[Y])}] \leq e^{\eta^2 c_\eta \mathbb{E}[Y^2] }.
\end{align}
By substituting $Y = -U$, we finally have,
\begin{align}
   \mathbb{E}[e^{\eta (\mathbb{E}[U] - U)}] \leq e^{\eta^2 c_\eta \mathbb{E}[U^2] },
\end{align}
which completes the proof.
\end{proof}
Using the Bernstein condition and we also assume that $r(w,z_i)$ is lower bounded by $-b$ almost surely, we have for all $0< \eta < \frac{1}{b}$ and all $c > \frac{e^{\eta b} - \eta b - 1}{\eta^2b^2} > 0$, the following inequality holds with Lemma~\ref{lemma:exp_bern}:
 \begin{align}
     e^{\eta(\Esub{P_{W}\otimes \mu}{r(W,Z_i)} - r(W,Z_i))} \leq e^{\eta^2 c \Esub{P_{W}\otimes \mu}{r^2(W,Z_i)}}. \label{eq:rsquare}
 \end{align}
 Then following Proposition 5 in \cite{grunwald2021pac}, we have that for some $\beta \in [0,1]$, all $c>0$, $\eta < \frac{1}{2Bc}$ and  all $0 < \beta' \leq \beta$, we have
\begin{align}
    &\eta^2c \Esub{P_W \otimes \mu}{r^2(w,Z_i)} \nonumber \\
    &\quad \leq \left(\frac{1}{2} \wedge \beta'\right) \eta \left(\mathbb{E}_{P_W\otimes \mu}[r(w,Z_i)]\right) \nonumber \\
    &\qquad + (1-\beta') \cdot(2 B c \eta)^{\frac{1}{1-\beta'}}\eta.
\end{align}
Hence the equation~(\ref{eq:rsquare}) can be further bounded by
\begin{align}
    &e^{\eta(\Esub{P_{W}\otimes \mu}{r(W,Z_i)} - r(W,Z_i))} \nonumber \\
    &\leq e^{\left(\frac{1}{2} \wedge \beta'\right) \eta \left({\mathbb{E}_{P_W \otimes \mu}}[r(w,Z_i)]\right)+(1-\beta')(2 B c \eta)^{\frac{1}{1-\beta'}}\eta } 
    \label{bound:bernstein}
\end{align}
for $\eta < \min(\frac{1}{2B(e-2)}, \frac{1}{b})$. Here we can choose $c$ to be $\max_{\eta} \frac{e^{\eta b} - \eta b - 1}{\eta^2b^2} = e-2$ since the function $\frac{e^x - x - 1}{x^2}$ is non-decreasing in $[0,1]$. Furthermore, if $\beta' = 1$, we can rewrite~(\ref{bound:bernstein}) as,
\begin{align}
    e^{\eta(\Esub{P_{W}\otimes \mu}{r(W,Z_i)} - r(W,Z_i))} \leq e^{\frac{1}{2} \eta  \left({\mathbb{E}_{P_W \otimes \mu}}[r(w,Z_i)]\right)} 
\end{align}
which completes the proof.
\end{proof}

\section{Proof of~Corollary~\ref{coro:central}}
\begin{proof}
We first present the following Lemma for bounding the excess risk using the cumulant generating function.
\begin{lemma}[Generalized from Lemma 13 in \cite{grunwald2020fast}]\label{lemma:central}
Let $\bar{\eta}>0$. Assume that the expected $\eta$-strong central condition holds, and suppose further that the $(u, c)$-witness condition holds for $u>0$ and $c \in(0,1]$. Let $c_{u}:=\frac{1}{c} \frac{\eta' u+1}{1-\frac{\eta'}{\eta}} > 1$, then the following inequality holds:
\begin{equation}
\mathbb{E}_{P_W \otimes \mu}\left[r(W,Z)\right]  \leq - \frac{c_{u}}{\eta'}  \log \mathbb{E}_{ P_W \otimes \mu}\left[e^{-\eta' r(W,Z)}\right] . 
\end{equation}

\end{lemma}
The proof of the above lemma is similar to the proof in Appendix C.1 (page 48) in \cite{grunwald2020fast} by taking the expectation over the hypothesis distribution $P_W$, which is omitted here. Now with~Lemma~\ref{lemma:central}, we have that for any $0 < \eta' < \eta$,
\begin{align}
    &\log \mathbb{E}_{ P_W \otimes \mu}\left[e^{-\eta' \left( r(W,Z) - \mathbb{E}_{P_W \otimes \mu}[r(W,Z)] \right)}\right] \nonumber \\
    &\quad \leq -\frac{\eta'}{c_u}\mathbb{E}_{P_W \otimes \mu}\left[r(W,Z)\right] \nonumber \\
    &\qquad + \eta'\mathbb{E}_{P_W \otimes \mu}\left[r(W,Z)\right] \nonumber \\
    &\quad = \left(1-\frac{1}{c_u}\right) \eta' \mathbb{E}_{P_W \otimes \mu}\left[r(W,Z)\right].
\end{align}
Therefore, the central condition with the witness condition implies the expected $(\eta', \frac{c-\frac{c\eta'}{\eta}}{\eta' u +1})$-central condition for any $0 < \eta' < \eta$.
\end{proof}

\section{Proof of Corollary \ref{coro:subexponential}} \label{proof:coro_subexponential_subgamma}

\begin{proof}[Proof of sub-Exponential Distribution]\label{proof:subexponential}
Recall that $X$ is a $\left(\nu^2, \alpha\right)$-sub-exponential random variable with parameters $\nu, \alpha>0$ if:
\begin{align}
\log\mathbb{E} \left[ e^{\lambda (X -\mathbb{E}[X])}\right] \leq \frac{\lambda^2 \nu^2}{2} , \quad \forall \lambda:|\lambda|<\frac{1}{\alpha}. 
\end{align}
Then rewrite the above inequality and replacing $\lambda = -\lambda$,  we have:
\begin{align}
    \log\mathbb{E} \left[ e^{-\lambda X}\right]
    &\leq -\lambda\left(\mathbb{E}[X] -\frac{\lambda \nu^2}{2}\right) \nonumber \\
    &= -\lambda \mathbb{E}[X]\left(1 - \frac{\lambda \nu^2}{2\mathbb{E}[X]}\right), 
\end{align}
for any $0 < \lambda<\frac{1}{\alpha}$. Let $c = \frac{1}{2}$, $\lambda = \min(\frac{1}{\alpha}, \frac{\nu^2}{\mathbb{E}_{P_W\otimes \mu}[r(W,Z)]})$ and we can complete the proof since $\mathbb{E}_{P_W\otimes \mu}[r(W,Z)] > 0$ from the definition.
\end{proof}

\begin{proof}[Proof of sub-Gamma distribution]
Recall that $-X$ is a ($\nu^2,\alpha$)-sub-Gamma random variable with variance parameter $\nu^2$ and scale parameter $\alpha$ if:
\begin{align}
\log\mathbb{E} \left[ e^{\lambda (\mathbb{E}[X] - X)}\right]\leq \frac{\nu^2 \lambda^2}{2(1- \alpha \lambda)}, \quad  \forall \lambda:  0<\lambda<\frac{1}{\alpha}.
\end{align}
Rewrite the above equation and we have:
\begin{align}
& \log\mathbb{E} \left[ e^{-\lambda X}\right] \leq \frac{\nu^2 \lambda^2}{2(1- \alpha \lambda)} - \lambda\mathbb{E}[X] \nonumber \\
& = -\lambda \mathbb{E}[X]\left(1 - \frac{\nu^2\lambda}{2(1-\alpha \lambda)\mathbb{E}[X]}\right), \nonumber \\
& \forall \lambda:  0<\lambda<\frac{1}{\alpha}.
\end{align}
Let $c = \frac{1}{2}$, $\lambda = \min(\frac{1}{\alpha}, \frac{\mathbb{E}_{P_W\otimes \mu}[r(W,Z)]}{\nu^2+\alpha \mathbb{E}_{P_W\otimes \mu}[r(W,Z)]})$ and we can complete the proof.
\end{proof}

\section{Proof of Corollary~\ref{coro:rerm}}\label{proof:rerm}
\begin{proof}
We first define,
\begin{align}
    \hat{{L}}_{\textup{reg}}(w,\mathcal{S}_n) := \hat{{L}}(w,\mathcal{S}_n) + \frac{\lambda}{n}g(w).
\end{align}
Based on Theorem~\ref{thm:eta-c}, we can bound the excess risk for $W_{\sf{RERM}}$ by (\ref{eq:start})-(\ref{eq:end}),
\begin{figure*}[!htb]
\normalsize
\begin{align}
     \mathbb{E}_{W}[\mathcal{R}(W_{\sf{RERM}})] & \leq  \frac{1}{c} \mathbb{E}_{P_{W\mathcal{S}_n}}[\hat{\mathcal{R}}\left(W_{\sf{RERM}}, \mathcal{S}_{n} \right)]  + \frac{1}{c\eta n} \sum_{i=1}^{n} I(W_{\sf{RERM}};Z_i) \label{eq:start}\\
     &= \frac{1}{c} \left( \mathbb{E}_{P_{W\mathcal{S}_n}}[\hat{L}\left(W_{\sf{RERM}}, \mathcal{S}_{n} \right) - \hat{L}\left(w^*, \mathcal{S}_{n} \right)] \right)   + \frac{1}{c\eta n} \sum_{i=1}^{n} I(W_{\sf{RERM}};Z_i) \\
     &\overset{(a)}{\leq} \frac{1}{c} \left( \mathbb{E}_{P_{W\mathcal{S}_n}}[\hat{{L}}_{\textup{reg}} \left(W_{\sf{RERM}}, \mathcal{S}_{n} \right)] - \mathbb{E}_{P_{W\mathcal{S}_n}}[\hat{{L}}_{\textup{reg}}\left(w^*, \mathcal{S}_{n}\right)] \right)  + \frac{\lambda B}{cn}+ \frac{1}{c\eta n} \sum_{i=1}^{n} I(W_{\sf{RERM}};Z_i) \\
     & =  \frac{1}{c} \mathbb{E}_{P_{W\mathcal{S}_n}}[\hat{\mathcal{R}}_{\textup{reg}}\left(W_{\sf{RERM}}, \mathcal{S}_{n} \right)] + \frac{\lambda B}{cn} + \frac{1}{c\eta n} \sum_{i=1}^{n} I(W_{\sf{RERM}};Z_i) \\
     &\overset{(b)}{\leq}  \frac{\lambda B}{cn} + \frac{1}{c\eta n} \sum_{i=1}^{n} I(W_{\sf{RERM}};Z_i). \label{eq:end}
 \end{align}
 \vspace*{4pt}
\hrulefill
\end{figure*}
where (a) follows since $|g(w^*) - g(W_{\sf{RERM}}))| \leq B$ the expected empirical risk is negative for $W_{\ERM}$ and (b) holds due to that $W_{\sf{RERM}}$ is the minimizer of the regularized loss. 
\end{proof}

\section{Proof of Theorem~\ref{lemma:intermediate}}\label{proof:intermediate}
\begin{proof}
We will build upon~(\ref{eq:MI_KL}). With the $(v,c)$-central condition, for any $\epsilon \geq 0$ and any $ 0 < \eta \leq v(\epsilon)$, the Jensen's inequality yields (\ref{eq:bound-v-central}).
\begin{figure*}[!htb]
\normalsize
\begin{align}
    \log \mathbb{E}_{P_{W}\otimes \mu}[e^{\eta( \mathbb{E}_{P_{W}\otimes \mu}[r(W,Z_i)]  - r(W,Z_i))  }]   &=  \log \mathbb{E}_{P_{W}\otimes \mu}[e^{\frac{\eta}{v(\epsilon)}v(\epsilon)( \mathbb{E}_{P_{W}\otimes \mu}[r(W,Z_i)]  - r(W,Z_i))  }] \\
    &\leq \log \left( \mathbb{E}_{P_{W}\otimes \mu}[e^{v(\epsilon)( \mathbb{E}_{P_{W}\otimes \mu}[r(W,Z_i)]  - r(W,Z_i))  }] \right)^{\frac{\eta}{v(\epsilon)}} \\
    &\leq \frac{\eta}{v(\epsilon)} \left( (1-c)v(\epsilon) \mathbb{E}_{P_{W}\otimes \mu}[r(W,Z_i)] + v(\epsilon) \epsilon \right) \\
    &= \eta(1-c) \mathbb{E}_{P_{W}\otimes \mu}[r(W,Z_i)] + \eta\epsilon.
    \label{eq:bound-v-central}
\end{align}
\vspace*{4pt}
\hrulefill
\end{figure*}
Substitute (\ref{eq:bound-v-central}) into (\ref{eq:MI_KL}), we arrive at,
\begin{align}
    I(W;Z_i) \geq &\eta \left(\Esub{P_W \otimes \mu}{r(W,Z_i)} - \Esub{P_{WZ_i}}{r(W,Z_i)}\right) \nonumber \\
    & - (1-c)\eta \mathbb{E}_{P_{W}\otimes \mu}[r(W,Z_i)] - \eta \epsilon.
\end{align}
Dividing $\eta$ on both sides, we arrive at,
\begin{align}
    \frac{I(W;Z_i)}{\eta} \geq &\Esub{P_W \otimes \mu}{r(W,Z_i)} - \Esub{P_{WZ_i}}{r(W,Z_i)} \nonumber \\
    & - (1-c) \mathbb{E}_{P_{W}\otimes \mu}[r(W,Z_i)] - \epsilon.
\end{align}
Rearranging the equation and it yields,
\begin{align}
    c\Esub{P_W \otimes \mu}{r(W,Z_i)} \leq & \Esub{P_{WZ_i}}{r(W,Z_i)} \nonumber \\
    & + \frac{I(W;Z_i)}{\eta} +\epsilon.
\end{align}
Therefore,
\begin{align}
    & \Esub{P_W \otimes \mu}{r(W,Z_i)} - \Esub{P_{WZ_i}}{r(W,Z_i)} \nonumber \\
    & \leq  (\frac{1}{c} - 1)\left({\mathbb{E}_{P_{WZ_i}}}[r(w,Z_i)]\right) + \frac{I(W;Z_i)}{c\eta} + \frac{\epsilon}{c}.
\end{align}
Summing up every term for $Z_i$ and dividing by $n$, we end up with,
\begin{align}
    &\Esub{P_W \otimes P_{\mathcal{S}_n}}{\hat{\mathcal{R}}(W,\mathcal{S}_n)} \nonumber  - \Esub{P_{W\mathcal{S}_n}}{\hat{\mathcal{R}}(W,\mathcal{S}_n)} \\
    &\leq \left(\frac{1}{c} - 1\right) \mathbb{E}_{P_{W\mathcal{S}_n}}[\hat{\mathcal{R}}(W,\mathcal{S}_n)] \nonumber \\
    &\quad + \frac{1}{n}\sum_{i=1}^{n}\left( \frac{I(W;Z_i)}{c\eta} + \frac{\epsilon}{c}\right).
\end{align}
Finally, we arrive at the following inequality:
\begin{align}
    \mathbb{E}_{W\mathcal{S}_n}[\mathcal{E}(W)] \leq &(\frac{1}{c} - 1) \left({\mathbb{E}_{P_{W\mathcal{S}_n}}}[\hat{\mathcal{R}}(W,\mathcal{S}_n)]\right) \nonumber  \\
    & + \frac{1}{n}\sum_{i=1}^{n}\left(\frac{I(W;Z_i)}{c\eta} + \frac{\epsilon}{c} \right).
\end{align}
In particular, if $v(\epsilon) = \epsilon^{1-\beta}$ for some $\beta \in [0,1]$, then by choosing $\eta = v(\epsilon)$ and $\frac{I(W;Z_i)}{c\eta} + \frac{\epsilon}{c}$ is optimized when $\epsilon = I(W;Z_i)^{\frac{1}{2-\beta}}$ and the bound becomes,
\begin{align}
    \mathbb{E}_{W\mathcal{S}_n}[\mathcal{E}(W)] \leq & (\frac{1}{c} - 1) \left({\mathbb{E}_{P_{W\mathcal{S}_n}}}[\hat{\mathcal{R}}(W,\mathcal{S}_n)]\right) \nonumber \\
    & + \frac{2}{nc}\sum_{i=1}^{n} I(W;Z_i)^{\frac{1}{2-\beta}},
\end{align}
which completes the proof.
\end{proof}

\section{Calculation Details of Example~\ref{sec:example}}\label{apd:example2}
In this section, we present the calculation details of the Gaussian mean estimation case. Let us consider the 1D-Gaussian mean estimation problem. Let $\ell(w,z_i) = (w-z_i)^2$; each sample is drawn from some Gaussian distribution, i.e., $Z_i \sim \mathcal{N}(\mu, \sigma_N^2)$. Then the ERM algorithm arrives at,
\begin{equation}
 W_{\ERM} = \frac{1}{n} \sum_{i=1}^{n} Z_i \sim \mathcal{N}(\mu, \frac{\sigma_N^2}{n}).
\end{equation}
It can be easily calculated that the optimal $w^*$ satisfies:
\begin{align}
    w^* & = \argmin \Esub{Z}{\ell(w,Z)} \nonumber \\
    & =\argmin \Esub{Z}{(w-Z)^2} \nonumber \\
    & = \mu.
\end{align}
Also, it can be calculated that the expected excess risk is,
\begin{align}
    \Esub{W}{\mathcal{R}(W_\ERM)} 
    &= \mathbb{E}_{P_W \otimes \mu}[\ell(W_\ERM,Z)] \nonumber \\
    &\qquad - \mathbb{E}_{Z}[\ell(w^*,Z)] \\
    &= \mathbb{E}_{P_W \otimes \mu}[(W_\ERM - Z)^2] \nonumber \\
    &\qquad - \mathbb{E}_{Z}[(\mu - Z)^2] \\
    &= \mu^2 + \frac{\sigma_N^2}{n} + \mu^2 + \sigma_N^2 - 2\mu^2 \nonumber \\
    &\qquad - \mu^2 - \mu^2 -\sigma_N^2 + 2\mu^2 \\
    &= \frac{\sigma_N^2}{n}.
\end{align}
The corresponding empirical excess risk is given by,
\begin{align}
    &\Esub{W\mathcal{S}_n}{\hat{\mathcal{R}}(W_\ERM,\mathcal{S}_n)} \nonumber \\
    & = \Esub{W_\ERM \mathcal{S}_n}{\hat L(W_\ERM,\mathcal{S}_n) - \hat{L}(w^*,\mathcal{S}_n)} \\
    &= \E{\frac{1}{n}\sum_{i=1}^{n}(W-Z_i)^2}  - \E{\frac{1}{n}\sum_{i=1}^{n}(\mu - Z_i)^2} \\
    &= \mu^2+ \frac{\sigma_N^2}{n} - \mu^2 - 2\mu^2 - \frac{2}{n} \sigma_N^2 + 2\mu^2 \\
    &= -\frac{\sigma_N^2}{n}.
\end{align}
Then it yields the expected generalization error as,
\begin{align}
    &\Esub{W\mathcal{S}_n}{\mathcal{E}(W_\ERM, \mathcal{S}_n)} \nonumber \\
    &= \Esub{W\mathcal{S}_n}{\mathcal{R}(W_\ERM)-\hat {\mathcal{R}}(W_\ERM,\mathcal{S}_n)} \\
    &= \frac{\sigma_N^2}{n} - (- \frac{\sigma_N^2}{n}) \\
    &= \frac{2\sigma_N^2}{n}.
\end{align}
The expected loss can be calculated as,
\begin{align}
    \mathbb{E}_{P_W \otimes \mu}[\ell(W_{\ERM},Z)] = \frac{n+1}{n}\sigma_N^2 := \sigma_W^2.
\end{align}
Let us verify the moment-generating functions for the squared loss. Since $\ell(W_\ERM,Z)$ is $\sigma^2_W\chi^2_1$ distributed, 
\begin{align}
    \log \mathbb{E}_{P_W \otimes \mu}[e^{\eta (W_\ERM - Z)^2}] = -\frac{1}{2}\log(1- 2\sigma^2_W \eta).
\end{align}
Hence, 
\begin{align}
    &\log \mathbb{E}_{P_W \otimes \mu}[e^{\eta \left((W_\ERM - Z)^2 - \mathbb{E}[(W_\ERM - Z)^2] \right)}] \nonumber \\
    & = -\frac{1}{2}\log(1- 2\sigma^2_W \eta) - \sigma^2_W\eta.
\end{align}
It is easy to prove that for any $x\leq 0$, 
\begin{align}
    -\frac{1}{2}\log(1-2x) - x \leq x^2,
\end{align}
which yields,
\begin{align}
    & \log \mathbb{E}_{P_W \otimes \mu}[e^{\eta \left((W_\ERM - Z)^2 - \mathbb{E}[(W_\ERM - Z)^2] \right)}] \nonumber \\
    & = -\frac{1}{2}\log(1- 2\sigma^2_W \eta) - \sigma^2_W\eta \\
    & \leq \sigma_W^4 \eta^2 
\end{align}
for any $\eta < 0$. Therefore, $\ell(W,Z)$ is $\sqrt{2\sigma^4_W}$-sub-Gaussian and we can only achieve the slow rate of $O(\sqrt{1/n})$ with the bound by \cite{bu2020tightening}.

\section{Calculation Details of Example~\ref{example:sub-Gaussian-2}}\label{apd:example3}
Now we introduce $w^*$ in the sequel as a comparison. For a given $\eta$ and $w$, we calculate the moment generating function for the term $ r(w,Z) =(w - Z)^2 - (w^* - Z)^2$ as follows. Using the known results for moment-generating function of Gaussian random variables, we have
\begin{align}
 &\mathbb{E}_{\mu}[e^{\eta r(w,Z)}] \nonumber \\
 & = \frac{1}{\sqrt{2\pi \sigma_N^2}} \int  e^{-\frac{(z-\mu)^2}{2\sigma_N^2}}  e^{\eta ((w - z)^2 - (w^* - z)^2)}  dz \\
 &=\operatorname{exp}\left( (2\eta^2\sigma_N^2 + \eta)(w-\mu)^2 \right).
\end{align}
Taking expectation over $W$ w.r.t. the ERM solution, and using the known results for the moment-generating function of the chi-square random variables,  we have:
\begin{align}
 &\mathbb{E}_{P_W \otimes \mu}[e^{\eta r(W,Z)}] \nonumber \\
 & = \Esub{W}{\operatorname{exp}\left( (2\eta^2\sigma_N^2 + \eta) (W-\mu)^2 \right)}  \\  &= \sqrt{\frac{n}{ n- (4\eta^2\sigma_N^4 + 2 \eta\sigma_N^2)}}.
\end{align}
Therefore for any $\eta \in \mathbb{R}$ and any $n >  \max\left\{\frac{(4\eta^2\sigma^4_N+\eta\sigma^2_N)(2\eta^2\sigma^4_N+\eta\sigma^2_N)}{\eta^2\sigma^4_N}, 4\eta^2\sigma_N^4 + 2 \eta\sigma_N^2\right\} $, we arrive at, 
\begin{align}
    &\log  \mathbb{E}_{P_W \otimes \mu}[e^{\eta (r(W,Z) - \mathbb{E}[r(W,Z)]}] \nonumber \\
    &= \frac{1}{2} \log \frac{n}{ n- (4\eta^2\sigma_N^4 + 2 \eta\sigma_N^2)} - \frac{\eta\sigma_N^2}{n} \\
    & \leq \frac{1}{2} \frac{4\eta^2\sigma_N^4 + 2 \eta\sigma_N^2}{n - 4\eta^2\sigma^4_N - 2\eta\sigma^2_N} - \frac{\eta\sigma_N^2}{n} \\
    & \leq \frac{4\eta^2\sigma_N^4}{n}. 
\end{align}
which is of the order $O(\frac{1}{n})$ and we used the fact that $\frac{a+b}{n - 2a - 2b} \leq \frac{2a + b}{n}$ for $n > \max\left\{\frac{(2a+b)(2a+2b)}{a}, 2a + 2b\right\}$ with $a > 0$. The mutual information can be calculated as,
\begin{align}
    I(W_\ERM;Z_i) &= h(W_\ERM) - h(W_\ERM|Z_i)\\
    &= \frac{1}{2}\log \frac{2\pi e \sigma_N^2}{n} - \frac{1}{2}\log\frac{2\pi e (n-1)\sigma_N^2 }{n^2} \\
    & = \frac{1}{2}\log \frac{n}{n-1} \\
    & = O(\frac{1}{n}).
\end{align}
We then summarize all the quantities of interest in Table~\ref{tab:calculation_details} for references.
\begin{table*}[!ht]
    \centering
    \begin{tabular}{c|c}
    \hline 
     Quantity   &  Values / Distribution \\
     \hline 
     $\mathcal{S}_n$   &  $\{Z_1,Z_2,\cdots,Z_n\}$  \\
     $Z_i$     &     $\mathcal{N}(\mu,\sigma_N^2)$ \\
      $\ell(w,z)$   &  $(w-z)^2$ \\
      $\hat{L}(w,\mathcal{S}_n)$  & $\frac{1}{n}\sum_{i=1}^{n}\ell(w,z_i)$  \\ 
      $L(w)$  &  $\Esub{Z}{\ell(w,Z)}$ \\
      $W_\ERM$  & $\mathcal{N}(\mu,\frac{\sigma_N^2}{n})$    \\
      $w^*$   & $\mu$    \\
      $r(w,z)$  & $(w-z)^2 - (w^* - z)^2$  \\
      $\mathcal{R}(w)$  & $L(w) - L(w^*)$   \\ 
      $\hat{\mathcal{R}}(w,\mathcal{S}_n)$  & $\hat{L}(w) - \hat{L}(w^*)$    \\
      $\mathcal{E}(w,\mathcal{S}_n)$  & $L(w) - \frac{1}{n}\sum_{i=1}^{n}\ell(w,z_i)$   \\
      $M_{Z}[r(w,Z)]$ & $-\frac{1}{\eta} \log \mathbb{E}_Z\left[e^{-\eta r(w,Z)}\right]$   \\
      $M_{P_W\otimes \mu}[r(w,Z)]$ & $-\frac{1}{\eta} \log \mathbb{E}_{P_W\otimes \mu}\left[e^{-\eta r(W,Z)}\right]$   \\
      \hline 
      \hline
      $\Esub{W\mathcal{S}_n}{\hat{\mathcal{R}}(W_\ERM,\mathcal{S}_n)}$   &  $-\frac{\sigma_N^2}{n}$   \\
      $\Esub{W}{\mathcal{R}(W_\ERM)}/\mathbb{E}_{P_W\otimes \mu}[r(W,Z)]$   &  $\frac{\sigma_N^2}{n}$   \\
      $\Esub{W\mathcal{S}_n}{\mathcal{E}(W_\ERM, \mathcal{S}_n)}$ & $\frac{2\sigma_N^2}{n}$   \\
      $\mathcal{R}(w)/\mathbb{E}_{Z}[r(w,Z)]$  &  $(w-\mu)^2$    \\
      $\mathbb{E}_{Z}[e^{\eta r(w,Z)}]$ & $\operatorname{exp}\left( (2\eta^2\sigma_N^2 +\eta )(w-\mu)^2 \right)$   \\
      $\mathbb{E}_{P_W\otimes \mu}[e^{\eta r(W,Z)}]$ & $\sqrt{\frac{n}{ n- (4\eta^2\sigma_N^4 + 2 \eta\sigma_N^2)}}$   \\
      $\mathbb{E}_{Z}[e^{-\eta r(w,Z)}]$ & $\operatorname{exp}\left( (2\eta^2\sigma_N^2 - \eta )(w-\mu)^2 \right)$ \\
      $\mathbb{E}_{P_W\otimes \mu}[e^{-\eta r(W,Z)}]$ & $\sqrt{\frac{n}{ n- (4\eta^2\sigma_N^4 - 2 \eta\sigma_N^2)}}$ \\
     $M_{Z}[r(w,Z)]$ & $(1-2\eta\sigma_N^2)(w-\mu)^2$  \\
      $M_{P_W\otimes \mu}[r(w,Z)]$ & $(1-2\eta\sigma_N^2)\frac{\sigma_N^2}{n}$ \\
      $\mathbb{E}_{Z}[r(w,Z)^2]$ & $(w - \mu)^4 + 4(w -\mu)^2\sigma_N^2$ \\
      $\mathbb{E}_{P_W\otimes \mu}[r(W,Z)^2]$ & $\frac{3\sigma_N^4}{n^2} + \frac{4\sigma_N^4}{n}$  \\
      $I(W;Z_i)$ & $\frac{1}{2}\log\frac{n}{n-1}$ \\
      \hline 
    \end{tabular}
    \caption{Summarized Quantities} \label{tab:calculation_details}
\end{table*}

\section{Proof of Lemma \ref{lemma:tightness}}\label{proof:lemma_tightness}
We are going to prove the lower bound as the upper bound is from the variational representations:
\begin{align*}
    \frac{n-1}{n} I(W;Z_i) 
    &\leq \mathbb{E}_{WZ_i}[-\eta r(W,Z_i)] \\
    &\qquad - \log \mathbb{E}_{P_W\otimes \mu}[e^{-\eta r(W,Z_i)}] \\
    &\leq I(W;Z_i).
\end{align*}
The proof for Lemma \ref{lemma:tightness} simply follows the calculation of the two terms. By setting $\eta = \frac{1}{2\sigma^2_N}$, the first term can be calculated as:
\begin{align}
    \mathbb{E}_{WZ_i}[-\eta r(W,Z_i)] = \frac{1}{2n},
\end{align}
while the second term could be calculated as:
\begin{align}
    \log \mathbb{E}_{P_W\otimes \mu}[e^{-\eta r(W,Z_i)}] = 0.
\end{align}
Hence we have:
\begin{align}
    &\frac{n-1}{n}I(W;Z_i)=\frac{n-1}{2n} \log\frac{n}{n-1} \\
    & \leq \frac{1}{2}\log\frac{n}{n-1} \\
    &= I(W;Z_i).
\end{align}

\section{Condition Checking} \label{apd:condiiton_checking}
From Table~\ref{tab:calculation_details}, we can check to conclude that for most fast rate conditions such as Berstein's condition, central condition, and sub-Gaussian condition, the results will hold in expectation, but this is not the case for any $w \in \mathcal{W}$. To see this, we will check whether the condition is in succession.
\begin{itemize}
    \item When checking $\eta$-central condition, 
    \begin{itemize}
    \item For any $w$, 
        \begin{align}
        & \mathbb{E}_{\mu}[e^{-\eta r(w,Z)}] \nonumber = \operatorname{exp}\left( (2\eta^2\sigma_N^2 -\eta )(w-\mu)^2 \right) \nonumber \\
        &\leq 1,
        \end{align}
        then we require $0 < \eta \leq \frac{1}{2\sigma_N^2}$.
        \item For $W_\ERM$,
        \begin{align}
        &\mathbb{E}_{P_W\otimes \mu}[e^{-\eta r(W,Z)}]  = \sqrt{\frac{n}{ n- (4\eta^2\sigma_N^4 - 2 \eta\sigma_N^2)}} \nonumber \\
        &\leq 1 .
        \end{align}
        then we require $0 < \eta \leq \frac{1}{2\sigma_N^2}$.
    \end{itemize}
    \item When checking Bernstein's condition,
    \begin{itemize}
        \item For any $w\in \mathcal{W}$,
        \begin{align}
          &\mathbb{E}_{\mu}[r(w,Z)^2] =  (w - \mu)^4 + 4(w -\mu)^2\sigma_N^2 \nonumber \\
          & \leq B(\mathbb{E}_{ Z}[r(w,Z)])^{\beta} =  B(w-\mu)^{2\beta}.
        \end{align}
        Apparently, this does not hold for all $w \in \mathbb{R}$ when $\beta \in [0,1]$.
        \item For $W_\ERM$,
        \begin{align}
        &\mathbb{E}_{P_W\otimes \mu}[r(W_\ERM,Z)^2]  = \frac{3\sigma_N^4}{n^2}+ \frac{4\sigma_N^4}{n} \nonumber \\
        &\leq B(\mathbb{E}_{P_W\otimes \mu}[r(W_\ERM,Z)])^{\beta} \nonumber \\
        &=  B(\frac{\sigma_N^2}{n})^{\beta}.
        \end{align}
        This holds for $\beta = 1$ and $B = 7\sigma_N^2$.
    \end{itemize}
    \item When checking witness condition,
    \begin{itemize}
    \item For any $w \in \mathcal{W}$, we require that,
    \begin{align}
        &\mathbb{E}_{Z}\left[\left(r(w,Z) \right) \cdot \mathbf{1}_{\left\{r(w,Z) \leq u \right\}}\right] \nonumber \\
        &\geq c \mathbb{E}_{Z}\left[r(w,Z) \right] = c(w-\mu)^2.
    \end{align}
    There does not exist finite $c$ and $u$ that satisfy the above inequality, so the witness condition does not hold for all $w \in \mathcal{W}$.
    \item For $W_\ERM$,
    \begin{align}
        & \mathbb{E}_{P_W\otimes \mu}\left[\left(r(W_\ERM,Z) \right) \cdot \mathbf{1}_{\left\{r(W_\ERM,Z) \leq u \right\}}\right] \nonumber \\
        & \geq c \mathbb{E}_{P_W\otimes \mu}\left[r(W,Z) \right] = \frac{c\sigma_N^2}{n}.
    \end{align}
    In this case, with high probability $r(W,Z)$ approaches zero, and there exists $u$ and $c$ satisfying the above inequality.
    \end{itemize}
    \item When checking the sub-Gaussian condition,
    \begin{itemize}
    \item For $W_\ERM$, when $0 < \eta \leq \frac{1}{2\sigma_N^2}$, we have,
    \begin{align}
        & \log \mathbb{E}_{P_W\otimes \mu}\left[e^{-\eta \left( r(W_\ERM,Z) - \mathbb{E}[r(W_\ERM,Z)]\right)} \right] \nonumber \\
        &= O\left(\frac{2\eta^2\sigma_N^4}{n}\right).
    \end{align}
    Then it satisfy with the $\sigma'^2$-sub-Gaussian condition that $\sigma'^2 = \frac{4\sigma_N^4}{n}$.
    \item For any $w$, 
     \begin{align}
          \log \mathbb{E}_{Z}[e^{\eta r(w,Z)}] = 2\eta^2\sigma_N^2(w-\mu)^2.
     \end{align}
     Since $w$ is unbounded, it does not satisfy the sub-Gaussian assumption for all $w\in\mathcal{W}$.
    \end{itemize}
    \item When checking the $(\eta,c)$-central condition, we confirmed that the learning tuple satisfies the $(\eta,c)$-central condition under $P_W\otimes \mu$ for the ERM algorithm. Now we consider the case for any hypothesis $w$.
    \begin{itemize}
        \item For any constant hypothesis $w$: we have that 
    \begin{align}
    \mathbb{E}_{\mu}[e^{-\eta r(w,Z)}]  = \operatorname{exp}\left( (2\eta^2\sigma_N^2 - \eta )(w-\mu)^2 \right),
    \end{align}
    where the excess risk can be calculated as:
    \begin{align}
    \mathbb{E}_{\mu}[r(w,Z)] = (w-\mu)^2.
    \end{align}
    Therefore, the $(\eta,c)$-central condition is satisfied as
    \begin{align}
    \log \mathbb{E}_{\mu}[e^{-\eta r(w,Z)}] &=(2\eta^2\sigma_N^2 - \eta )(w-\mu)^2 \nonumber \\
    & \leq -c\eta (w-\mu)^2.
    \end{align}
    for any $\eta \leq \frac{1}{2\sigma^2_N}$ and any $c \leq 1 - 2\eta \sigma^2_N$. 
    \end{itemize}
\end{itemize}

\newpage 

\bibliographystyle{IEEEtranN}
\bibliography{reference}

\begin{thebibliography}{42}
\providecommand{\natexlab}[1]{#1}
\providecommand{\url}[1]{#1}
\csname url@samestyle\endcsname
\providecommand{\newblock}{\relax}
\providecommand{\bibinfo}[2]{#2}
\providecommand{\BIBentrySTDinterwordspacing}{\spaceskip=0pt\relax}
\providecommand{\BIBentryALTinterwordstretchfactor}{4}
\providecommand{\BIBentryALTinterwordspacing}{\spaceskip=\fontdimen2\font plus
\BIBentryALTinterwordstretchfactor\fontdimen3\font minus \fontdimen4\font\relax}
\providecommand{\BIBforeignlanguage}[2]{{%
\expandafter\ifx\csname l@#1\endcsname\relax
\typeout{** WARNING: IEEEtranN.bst: No hyphenation pattern has been}%
\typeout{** loaded for the language `#1'. Using the pattern for}%
\typeout{** the default language instead.}%
\else
\language=\csname l@#1\endcsname
\fi
#2}}
\providecommand{\BIBdecl}{\relax}
\BIBdecl

\bibitem[Wu et~al.(2022)Wu, Manton, Aickelin, and Zhu]{wu2022fast}
X.~Wu, J.~H. Manton, U.~Aickelin, and J.~Zhu, ``Fast rate generalization error bounds: Variations on a theme,'' in \emph{2022 IEEE Information Theory Workshop (ITW)}.\hskip 1em plus 0.5em minus 0.4em\relax IEEE, 2022, pp. 43--48.

\bibitem[Bu et~al.(2020)Bu, Zou, and Veeravalli]{bu2020tightening}
Y.~Bu, S.~Zou, and V.~V. Veeravalli, ``Tightening mutual information-based bounds on generalization error,'' \emph{IEEE Journal on Selected Areas in Information Theory}, vol.~1, no.~1, pp. 121--130, 2020.

\bibitem[Russo and Zou(2016)]{russo2016controlling}
D.~Russo and J.~Zou, ``Controlling bias in adaptive data analysis using information theory,'' in \emph{Artificial Intelligence and Statistics}.\hskip 1em plus 0.5em minus 0.4em\relax PMLR, 2016, pp. 1232--1240.

\bibitem[Xu and Raginsky(2017)]{xu2017information}
A.~Xu and M.~Raginsky, ``Information-theoretic analysis of generalization capability of learning algorithms,'' in \emph{Advances in Neural Information Processing Systems}, 2017, pp. 2521--2530.

\bibitem[Vapnik and Chervonenkis(1971)]{vapnik1971uniform}
V.~N. Vapnik and A.~Y. Chervonenkis, ``On the uniform convergence of relative frequencies of events to their probabilities,'' \emph{Theory of Probability and Its Applications}, 1971.

\bibitem[Kearns and Ron(1997)]{kearns1997algorithmic}
M.~Kearns and D.~Ron, ``Algorithmic stability and sanity-check bounds for leave-one-out cross-validation,'' in \emph{Proceedings of the tenth annual conference on Computational learning theory}, 1997, pp. 152--162.

\bibitem[Devroye and Wagner(1979)]{devroye1979distributiona}
L.~Devroye and T.~Wagner, ``Distribution-free inequalities for the deleted and holdout error estimates,'' \emph{IEEE Transactions on Information Theory}, vol.~25, no.~2, pp. 202--207, 1979.

\bibitem[Bousquet and Elisseeff(2002)]{bousquet2002stability}
O.~Bousquet and A.~Elisseeff, ``Stability and generalization,'' \emph{Journal of Machine Learning Research}, vol.~2, no. Mar, pp. 499--526, 2002.

\bibitem[Raginsky et~al.(2016)Raginsky, Rakhlin, Tsao, Wu, and Xu]{raginsky2016information}
M.~Raginsky, A.~Rakhlin, M.~Tsao, Y.~Wu, and A.~Xu, ``Information-theoretic analysis of stability and bias of learning algorithms,'' in \emph{IEEE Information Theory Workshop}.\hskip 1em plus 0.5em minus 0.4em\relax IEEE, 2016, pp. 26--30.

\bibitem[Steinke and Zakynthinou(2020)]{steinke2020reasoning}
T.~Steinke and L.~Zakynthinou, ``Reasoning about generalization via conditional mutual information,'' in \emph{Conference on Learning Theory}.\hskip 1em plus 0.5em minus 0.4em\relax PMLR, 2020, pp. 3437--3452.

\bibitem[Asadi et~al.(2018)Asadi, Abbe, and Verd{\'u}]{asadi2018chaining}
A.~Asadi, E.~Abbe, and S.~Verd{\'u}, ``Chaining mutual information and tightening generalization bounds,'' \emph{Advances in Neural Information Processing Systems}, vol.~31, 2018.

\bibitem[McAllester(1999)]{mcallester1999some}
D.~A. McAllester, ``Some {PAC-B}ayesian theorems,'' \emph{Machine Learning}, vol.~37, no.~3, pp. 355--363, 1999.

\bibitem[Alquier et~al.(2024)]{alquier2024user}
P.~Alquier \emph{et~al.}, ``User-friendly introduction to pac-bayes bounds,'' \emph{Foundations and Trends{\textregistered} in Machine Learning}, vol.~17, no.~2, pp. 174--303, 2024.

\bibitem[Hellstr{\"o}m and Durisi(2020)]{hellstrom2020generalization}
F.~Hellstr{\"o}m and G.~Durisi, ``Generalization bounds via information density and conditional information density,'' \emph{IEEE Journal on Selected Areas in Information Theory}, vol.~1, no.~3, pp. 824--839, 2020.

\bibitem[Negrea et~al.(2019)Negrea, Haghifam, Dziugaite, Khisti, and Roy]{negrea2019information}
J.~Negrea, M.~Haghifam, G.~K. Dziugaite, A.~Khisti, and D.~M. Roy, ``Information-theoretic generalization bounds for sgld via data-dependent estimates,'' \emph{Advances in Neural Information Processing Systems}, vol.~32, 2019.

\bibitem[Rodr{\'\i}guez-G{\'a}lvez et~al.(2021)Rodr{\'\i}guez-G{\'a}lvez, Bassi, Thobaben, and Skoglund]{rodriguez2021random}
B.~Rodr{\'\i}guez-G{\'a}lvez, G.~Bassi, R.~Thobaben, and M.~Skoglund, ``On random subset generalization error bounds and the stochastic gradient langevin dynamics algorithm,'' in \emph{2020 IEEE Information Theory Workshop (ITW)}.\hskip 1em plus 0.5em minus 0.4em\relax IEEE, 2021, pp. 1--5.

\bibitem[Zhou et~al.(2022)Zhou, Tian, and Liu]{zhou2022individually}
R.~Zhou, C.~Tian, and T.~Liu, ``Individually conditional individual mutual information bound on generalization error,'' \emph{IEEE Transactions on Information Theory}, 2022.

\bibitem[Haghifam et~al.(2020)Haghifam, Negrea, Khisti, Roy, and Dziugaite]{haghifam2020sharpened}
M.~Haghifam, J.~Negrea, A.~Khisti, D.~M. Roy, and G.~K. Dziugaite, ``Sharpened generalization bounds based on conditional mutual information and an application to noisy, iterative algorithms,'' \emph{Advances in Neural Information Processing Systems}, vol.~33, pp. 9925--9935, 2020.

\bibitem[Hellstr{\"o}m and Durisi(2021)]{hellstrom2021fast}
F.~Hellstr{\"o}m and G.~Durisi, ``Fast-rate loss bounds via conditional information measures with applications to neural networks,'' in \emph{2021 IEEE International Symposium on Information Theory (ISIT)}.\hskip 1em plus 0.5em minus 0.4em\relax IEEE, 2021, pp. 952--957.

\bibitem[Hellstr{\"o}m and Durisi(2022)]{hellstrom2022new}
------, ``A new family of generalization bounds using samplewise evaluated cmi,'' \emph{Advances in Neural Information Processing Systems}, vol.~35, pp. 10\,108--10\,121, 2022.

\bibitem[Zhou et~al.(2023)Zhou, Tian, and Liu]{zhou2023stochastic}
R.~Zhou, C.~Tian, and T.~Liu, ``Stochastic chaining and strengthened information-theoretic generalization bounds,'' \emph{Journal of the Franklin Institute}, vol. 360, no.~6, pp. 4114--4134, 2023.

\bibitem[Zhou et~al.(2024)Zhou, Tian, and Liu]{zhou2024exactly}
------, ``Exactly tight information-theoretic generalization error bound for the quadratic gaussian problem,'' \emph{IEEE Journal on Selected Areas in Information Theory}, 2024.

\bibitem[Van~Erven et~al.(2015)Van~Erven, Gr{\"u}nwald, Mehta, Reid, and Williamson]{van2015fast}
T.~Van~Erven, P.~D. Gr{\"u}nwald, N.~A. Mehta, M.~D. Reid, and R.~C. Williamson, ``Fast rates in statistical and online learning,'' \emph{The Journal of Machine Learning Research}, vol.~16, no.~1, pp. 1793--1861, 2015.

\bibitem[Gr{\"u}nwald and Mehta(2020)]{grunwald2020fast}
P.~D. Gr{\"u}nwald and N.~A. Mehta, ``Fast rates for general unbounded loss functions: from erm to generalized bayes,'' \emph{Journal of Machine Learning Research}, vol.~21, no.~56, pp. 1--80, 2020.

\bibitem[Grunwald et~al.(2021)Grunwald, Steinke, and Zakynthinou]{grunwald2021pac}
P.~Grunwald, T.~Steinke, and L.~Zakynthinou, ``Pac-bayes, mac-bayes and conditional mutual information: Fast rate bounds that handle general vc classes,'' in \emph{Conference on Learning Theory}.\hskip 1em plus 0.5em minus 0.4em\relax PMLR, 2021, pp. 2217--2247.

\bibitem[Bartlett and Mendelson(2006)]{bartlett2006empirical}
P.~L. Bartlett and S.~Mendelson, ``Empirical minimization,'' \emph{Probability Theory and Related Fields}, vol. 135, no.~3, pp. 311--334, 2006.

\bibitem[Jiao et~al.(2017)Jiao, Han, and Weissman]{jiao2017dependence}
J.~Jiao, Y.~Han, and T.~Weissman, ``Dependence measures bounding the exploration bias for general measurements,'' in \emph{IEEE International Symposium on Information Theory (ISIT)}.\hskip 1em plus 0.5em minus 0.4em\relax IEEE, 2017, pp. 1475--1479.

\bibitem[Donsker and Varadhan(1975)]{donsker1975asymptotic}
M.~D. Donsker and S.~S. Varadhan, ``Asymptotic evaluation of certain markov process expectations for large time, i,'' \emph{Communications on Pure and Applied Mathematics}, vol.~28, no.~1, pp. 1--47, 1975.

\bibitem[Mehta(2017)]{mehta2017fast}
N.~Mehta, ``Fast rates with high probability in exp-concave statistical learning,'' in \emph{Artificial Intelligence and Statistics}.\hskip 1em plus 0.5em minus 0.4em\relax PMLR, 2017, pp. 1085--1093.

\bibitem[Koren and Levy(2015)]{koren2015fast}
T.~Koren and K.~Levy, ``Fast rates for exp-concave empirical risk minimization,'' \emph{Advances in Neural Information Processing Systems}, vol.~28, 2015.

\bibitem[Mhammedi et~al.(2019)Mhammedi, Gr{\"u}nwald, and Guedj]{mhammedi2019pac}
Z.~Mhammedi, P.~Gr{\"u}nwald, and B.~Guedj, ``Pac-bayes un-expected bernstein inequality,'' \emph{Advances in Neural Information Processing Systems}, vol.~32, 2019.

\bibitem[Zhu(2020)]{zhu2020semi}
J.~Zhu, ``Semi-supervised learning: The case when unlabeled data is equally useful,'' in \emph{Conference on Uncertainty in Artificial Intelligence}.\hskip 1em plus 0.5em minus 0.4em\relax PMLR, 2020, pp. 709--718.

\bibitem[Hafez-Kolahi et~al.(2020)Hafez-Kolahi, Golgooni, Kasaei, and Soleymani]{hafez2020conditioning}
H.~Hafez-Kolahi, Z.~Golgooni, S.~Kasaei, and M.~Soleymani, ``Conditioning and processing: Techniques to improve information-theoretic generalization bounds,'' \emph{Advances in Neural Information Processing Systems}, vol.~33, pp. 16\,457--16\,467, 2020.

\bibitem[Bartlett et~al.(2006)Bartlett, Jordan, and McAuliffe]{bartlett2006convexity}
P.~L. Bartlett, M.~I. Jordan, and J.~D. McAuliffe, ``Convexity, classification, and risk bounds,'' \emph{Journal of the American Statistical Association}, vol. 101, no. 473, pp. 138--156, 2006.

\bibitem[Hanneke(2016)]{hanneke2016refined}
S.~Hanneke, ``Refined error bounds for several learning algorithms,'' \emph{The Journal of Machine Learning Research}, vol.~17, no.~1, pp. 4667--4721, 2016.

\bibitem[Wong and Shen(1995)]{wong1995probability}
W.~H. Wong and X.~Shen, ``Probability inequalities for likelihood ratios and convergence rates of sieve {MLE}s,'' \emph{The Annals of Statistics}, pp. 339--362, 1995.

\bibitem[Borjesson and Sundberg(1979)]{borjesson1979simple}
P.~Borjesson and C.-E. Sundberg, ``Simple approximations of the error function $q(x)$ for communications applications,'' \emph{IEEE Transactions on Communications}, vol.~27, no.~3, pp. 639--643, 1979.

\bibitem[Gao et~al.(2017)Gao, Kannan, Oh, and Viswanath]{gao2017estimating}
W.~Gao, S.~Kannan, S.~Oh, and P.~Viswanath, ``Estimating mutual information for discrete-continuous mixtures,'' \emph{Advances in neural information processing systems}, vol.~30, 2017.

\bibitem[Moddemeijer(1989)]{moddemeijer1989estimation}
R.~Moddemeijer, ``On estimation of entropy and mutual information of continuous distributions,'' \emph{Signal Processing}, vol.~16, no.~3, pp. 233--248, 1989.

\bibitem[Kraskov et~al.(2004)Kraskov, St{\"o}gbauer, and Grassberger]{kraskov2004estimating}
A.~Kraskov, H.~St{\"o}gbauer, and P.~Grassberger, ``Estimating mutual information,'' \emph{Physical review E}, vol.~69, no.~6, p. 066138, 2004.

\bibitem[Boucheron et~al.(2013)Boucheron, Lugosi, and Massart]{boucheron2013concentration}
S.~Boucheron, G.~Lugosi, and P.~Massart, \emph{Concentration inequalities: A nonasymptotic theory of independence}.\hskip 1em plus 0.5em minus 0.4em\relax OUP Oxford, Feb. 2013.

\bibitem[Cesa-Bianchi and Lugosi(2006)]{cesa2006prediction}
N.~Cesa-Bianchi and G.~Lugosi, \emph{Prediction, learning, and games}.\hskip 1em plus 0.5em minus 0.4em\relax Cambridge University Press, 2006.

\end{thebibliography}

\newpage 

\begin{IEEEbiographynophoto}{Xuetong Wu}
Xuetong Wu received a B.S. degree in 2016 from Fudan University, Shanghai, China, and completed his M.S. and Ph.D. degrees in 2018 and 2023, respectively, in the Department of Electrical and Electronic Engineering at the University of Melbourne. His research interests include machine learning, transfer learning, and information theory.
\end{IEEEbiographynophoto}
\begin{IEEEbiographynophoto}{Jonathan H. Manton}
Professor Jonathan Manton holds a Distinguished Chair at the University of Melbourne with the title Future Generation Professor. He is also an adjunct professor in the Mathematical Sciences Institute at the Australian National University, a Fellow of the Australian Mathematical Society (FAustMS) and a Fellow of the Institute of Electrical and Electronics Engineers (FIEEE). He received his Bachelor of Science (mathematics) and Bachelor of Engineering (electrical) degrees in 1995 and his Ph.D. degree in 1998, all from The University of Melbourne, Australia. In 2005 he became a full Professor in the Research School of Information Sciences and Engineering (RSISE) at the Australian National University. From mid-2006 till mid-2008, he was on secondment to the Australian Research Council as Executive Director, Mathematics, Information and Communication Sciences. He has served as an Associate Editor for the IEEE Transactions on Signal Processing and a Lead Guest Editor for the IEEE Transactions on Selected Topics in Signal Processing. He has been a Committee Member of the IEEE Signal Processing for Communications (SPCOM) Technical Committee, and a Committee Member on the Mathematics Panel for the ACT Board of Senior Secondary Studies in Australia. Currently he is a Committee Member of the IEEE Machine Learning for Signal Processing (MLSP) Technical Committee, and is the Signal Processing Chapter Chair for the IEEE Victorian Section. Awards include a prestigious Queen Elizabeth II Fellowship and a Future Summit Australian Leadership Award.

His principal fields of interest are Mathematical Systems Theory (including Signal Processing and Optimisation), Geometry and Topology (Differential and Algebraic), and Learning and Computation (including Systems Biology, Systems Neuroscience and Machine Learning).
\end{IEEEbiographynophoto}
\begin{IEEEbiographynophoto}{Uwe Aickelin}
Prof Uwe Aickelin received the Ph.D. degree from the University of Wales, U.K. He is currently a Professor and the Head of the School of Computing and Information Systems, University of Melbourne. His current research interests include artificial intelligence (modelling and simulation), data mining and machine learning (robustness and uncertainty), decision support and optimization (medicine and digital economy), and health informatics (electronic healthcare records).
\end{IEEEbiographynophoto}
\begin{IEEEbiographynophoto}{Jingge Zhu}
Jingge Zhu received the B.S. degree and M.S. degree in electrical engineering from Shanghai Jiao Tong University, Shanghai, China, the Dipl.-Ing. degree in Technische Informatik from Technische Universit\"{a}t Berlin, Berlin, Germany and the Doctorat \`{e}s Sciences degree from the Ecole Polytechnique F\'{e}d\'{e}rale (EPFL), Lausanne, Switzerland. He was a post-doctoral researcher at the University of California, Berkeley.  He is now a Senior Lecturer at the University of Melbourne, Australia. His research interests include information theory with applications in communication systems and machine learning. 

Dr. Zhu received the Discovery Early Career Research Award (DECRA) from the Australian Research Council in 2021, the IEEE Heinrich Hertz Award for Best Communications Letters in 2013, and the Early Postdoc. Mobility Fellowship from Swiss National Science Foundation in 2015 and the Chinese Government Award for Outstanding Students Abroad in 2016.
\end{IEEEbiographynophoto}

\end{document}